\documentclass[a4paper,11pt,leqno]{article}
\usepackage{amsmath,amsfonts,amsthm,amssymb}
\usepackage[left=3cm, right=3cm, top=3cm, bottom=3.5cm]{geometry}
\usepackage{graphicx}
\catcode`\@=11
\@addtoreset{equation}{section}

\newtheorem{theo}{Theorem}[section]
\newtheorem{prop}{Proposition}[section]

\newtheorem{lem}{Lemma}[section]
\newtheorem{rem}{Remark}[section]

\newcommand{\N}{{\mathbf{N}}}

\newcommand{\RR}{{\mathcal{R}}}

\title{Stationary motion of a self-gravitating toroidal incompressible liquid layer}
\author{Giorgio Fusco,\footnote{Dipartimento di Matematica Pura ed Applicata, Universit\`a degli Studi dell'Aquila, Via Vetoio, 67010 Coppito, L'Aquila, Italy; e-mail:{\texttt{fusco@univaq.it}}}\ \  Piero Negrini\footnote{Dipartimento di Matematica,   Sapienza, Universit\`a di Roma, Piazzale Aldo Moro 5, 00185 Rome, Italy; e-mail: {\texttt{negrini@mat.uniroma1.it}}}\ \ and Waldyr M. Oliva \footnote{CAMGSD and ISR, Instituto Superior Técnico, UTL - Lisboa;e-mail:{\texttt{wamoliva@math.ist.utl.pt
}}} }
\date{}
\begin{document}

\maketitle

\begin{abstract}
We consider an incompressible fluid contained in a toroidal stratum which is only subjected to Newtonian self-attraction. Under the assumption of infinitesimal tickness of the stratum we show the existence of stationary motions during which the stratum is approximatly a round torus (with radii $r,\,R$ and $R>>r$) that rotates around its axis and at the same time {\it rolls} on itself. Therefore each particle of the stratum describes an helix-like trajectory around the circumference of radius $R$ that connects the centers of the cross sections of the torus.
\end{abstract}
\vskip0.2cm
\noindent

\section{Introduction}
\label{introduction}
 The problem of the determination of the figures of equilibrium of a self-gravitating rotating mass has received a lot of attention in the classical literature beginning with the work of Newton on the oblateness of the earth \cite{chandra2},\cite{newton9} and \cite{oliva2}. The method of {\it canals} introduced by Newton was exploited by Maclaurin and Jacobi who discovered several families of ellipsoids of equilibrium. The search for stationary motions in ellipsoidal regions, including ellipsoids of equilibrium, of self-gravitating incompressible fluids was completed by Riemann \cite{riemann} who developed ideas of Dirichlet \cite{dirichlet}. Poincar\'e \cite{p} adopted a global point of view and introduced the concept of {\it bifurcation} in the attempt of describing the whole set of figures of equilibrium of rotating self-gravitating masses. More recently several papers have appeared approaching this classical problem with variational techniques \cite{a} \cite{ab} \cite{cf}. The case of a rotating solid torus $\mathcal{T}$ was first considered by Poincar\'e \cite{p}, see also \cite{sofie}, and revisited in \cite{a}. It is natural to expect that the motion studied in these papers belongs to a family of stationary motions where, beside rotating around its axis,  the torus rolls on itself so that the fluid particles describe a helix-like path around the circumference $\mathcal{C}$ of radius $R$ that connects the centers of the cross sections of the torus. We can also conjecture that beside this class of stationary motions of solid self-gravitating torii, there is also the possibility of similar stationary motions of self-gravitating toroidal strata. This last class of motions however will not contain as a special case the case of relative equilibrium since for a toroidal stratum the pressure cannot compensate the self attraction that tries to collapse the torus to the circumference $\mathcal{C}$. Assuming that this set of motions does exist, then rigid helicoidal motions of {\it cylindrical} strata should be in the {\it closure} of the set in the sense that, in the singular limit $R\rightarrow+\infty$, the motion of the toroidal stratum should converge to the motion of a cylindrical stratum.
 In this note we adopt this point of view and focus on the case of an incompressible self-gravitating toroidal stratum $\mathcal{T}_R$ of very small (infinitesimal) thickness and prove the existence of stationary motions of $\mathcal{T}_R$ of the type alluded to above for $R>>1$ (cfr. Theorem \ref{main}).

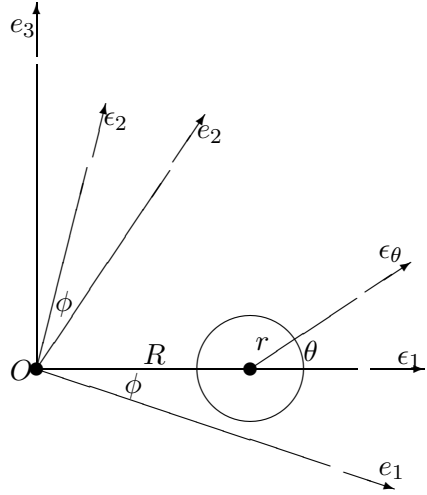
\begin{figure}[t]\label{geometry}

\begin{center}

\setlength{\unitlength}{1pt}
\begin{picture}(200,160) (-160,140)

 {\put(-100,178){\line(1,0){120}}}
 \put(25,178){\vector(1,0){20}}
 \put(35,180){$\epsilon_1$}
 {\put(-20,178){\line(3,2){40}}}
 \put(23,207){\vector(3,2){17}}
 \put(28,220){$\epsilon_\theta$}
 {\put(-100,178){\line(3,-1){110}}}
 \put(15,139){\vector(3,-1){19}}
 \put(28,138){$e_1$}
 {\put(-100,178){\line(0,1){115}}}
 \put(-100,296){\vector(0,1){20}}
 \put(-110,305){$e_3$}
 {\put(-100,178){\line(1,4){19}}}
 \put(-79,258){\vector(1,4){5}}
 \put(-75,270){$\epsilon_2$}
{\put(-100,178){\line(2,3){50}}}
 \put(-49,256){\vector(2,3){12}}
 \put(-40,265){$e_2$}
 \put(-20,178){\circle{40}}
 \put(-20,178){\circle*{5}}
 \put(-60,180){$R$}
 \put(-18,185){$r$}
 \put(0,181){$\theta$}
 \put(-67,169){$\phi$}
 \put(-93,200){$\phi$}
 \put(-100,178){\circle*{5}}
 \put(-110,174){$O$}

\end{picture}
\end{center}
\caption{The geometry of ${\mathcal T}_R$.}
\end{figure}

We represent $\mathcal{T}_R$ in the form, see Figure\ref{geometry}
 \begin{eqnarray}\label{torus}\left\{\begin{array}{lll}
   \hskip.5cm{\mathcal T}_R&=&\{x\in\RR^3:x=R\epsilon_1(\phi)+ (r(\theta)+\lambda)\epsilon_r(\theta,\phi), \\\nonumber
   && (\theta,\phi)\in S^1\times S^1, \lambda\in[-\frac{s(\theta)}{2},\frac{s(\theta)}{2}]\} \\\nonumber
    \hskip.2cm\epsilon_1(\phi)&=& \cos\phi e_1+\sin\phi e_2, \\\nonumber
    \epsilon_r(\theta,\phi)&=& \cos\theta\epsilon_1(\phi)+\sin\theta e_3,
     \end{array}
\right.
 \end{eqnarray}
where $e_j, j=1,2,3$ are the vector of the standard basis of $\RR^3$ , $\theta\rightarrow r(\theta)\simeq r_0$ is the polar representation of the cross section of the {\it middle fiber} ${\mathcal T}_R^0=
\{x\in{\mathcal T}_R: \lambda=0\}$ of ${\mathcal T}_R$
and $s(\theta)\simeq s_0$ is the thickness of the stratum. We let $r_0$ the average of $r(\theta)$ and consider the case where $\varepsilon:=\frac{r_0}{R}<<1$. Under this assumption we regard the stationary motion of the toroidal stratum as a perturbation of the limit helicoidal motion of a cylindrical stratum that we can ideally associate to $\varepsilon=0$.  Since
 cylinder and torus are different topological objects, the problem of continuing the motion of the fluid in the cylindrical stratum into a stationary motion in the toroidal stratum ${\mathcal T}_R$ is a singular perturbation problem.
  From a mathematical point of view the singularity cylinder-torus manifest itself in the fact that $r(\theta)$ and
  $s(\theta)$ and the other functions $\omega(\theta), \Omega^2(\theta)$ that we introduce
  to describe the velocity field on ${\mathcal T}_R$ can not be expanded in powers of $\varepsilon $ but terms of the form $\varepsilon^k\log{\frac{1}{\varepsilon}}$, $k\in \N$, must be included.

  Set \begin{eqnarray}\label{torusequation}
        X(\theta,\phi) &=& R\epsilon_1(\phi)+  r(\theta)\epsilon_r(\theta,\phi).
      \end{eqnarray}
  Under the standing assumption of infinitesimal thickness the velocity field
  $v=v(\theta,\phi,\lambda)$ on  ${\mathcal T}_R$ does not depend on $\lambda$ and can be
  computed on the middle fiber ${\mathcal T}_R^0=\{x=X(\theta,\phi), (\theta,\phi)\in S^1\times S^1\}$.
  Therefore we have
  $v=X_\theta\dot\theta +X_\phi\dot\phi$ where subscripts denote
  partial differentiation and $\dot{}$ time differentiation. Since we look for stationary motions which are
  invariant under rotations around the symmetry axis of ${\mathcal T}_R$ we have
  \begin{eqnarray}\label{omegaeomega}\left\{\begin{array}{lll}
  \dot\theta&=&\omega(\theta),\\
   \dot\phi &=&\Omega(\theta),
  \end{array}
\right.
 \end{eqnarray}
 for some $2\pi-$periodic functions $\Omega, \omega$ and therefore the velocity and acceleration
 vector fields on ${\mathcal T}_R^0$ are
  \begin{eqnarray}\label{v-a}\left\{\begin{array}{lll}
  v&=&X_\theta\omega+X_\phi\Omega,\\
   a&=&v_\theta\omega+v_\phi\Omega.
  \end{array}
\right.
 \end{eqnarray}
 From (\ref{torusequation}) that implies
 \begin{eqnarray}\label{xthetaxphi}
 X_\theta=(r\epsilon_r)_\theta,\hskip.3cm X_\phi=(R+r\cos\theta)\epsilon_2, \hskip.3cm X_\theta\cdot X_\phi=0,
 \end{eqnarray} and (\ref{omegaeomega})(\ref{v-a}) and a routine computation we get
 \begin{eqnarray}\label{velocity}
   v(\theta,\phi) &=&\omega(r\epsilon_r)_\theta+\Omega(R+r\cos\theta)\epsilon_2\\\nonumber
   &=&\omega(r^\prime\epsilon_r+ r\epsilon_\theta)+\Omega(R+r\cos\theta)\epsilon_2,
 \end{eqnarray}
and
 \begin{eqnarray}
\label{acceleration}
a(\theta, \phi) &=& [(r'' - r)\omega^2 + r' \omega' \omega
- (R + r\cos \theta)\cos \theta \Omega^2]\epsilon_r\nonumber\\
 &+ &[((R + r\cos \theta) \Omega)'\omega + (r \cos \theta)'\omega \Omega]\epsilon_2\nonumber\\
 &+ &[2 r'\omega^2 + r \omega' \omega + (R + r\cos \theta)\sin\theta \Omega^2]\epsilon_\theta,
\end{eqnarray}
where $\epsilon_r$ is defined in (\ref{torus}) and $\epsilon_\theta=-\sin\theta\epsilon_1+\cos\theta e_3, \epsilon_2=-\sin\phi e_1+\cos\phi e_2.$
Let $n=n(\theta,\phi)=\frac{X_\phi\wedge X_\theta}{\vert X_\phi\wedge X_\theta\vert}$ the exterior unit   normal
to ${\mathcal T}_R^0$ at $X(\theta,\phi)$ and $\delta=\delta(\theta)$ the thickness of the stratum along $n$
that is $\delta=s n\cdot\epsilon_r.$ The balance between inertial forces and newtonian self-attraction at the typical point $X(\phi,\theta)$ of ${\mathcal T}_R^0$ reads
\begin{eqnarray}\label{equilibrium}
  a(\phi,\theta) &=& G\mu\int_{S^1\times S^1}^*\frac{X(\alpha,\beta)-X(\theta,\phi)}
  {\vert X(\alpha,\beta)-X(\theta,\phi)\vert^3}\vert X_\phi\wedge X_\theta \vert(\alpha,\beta)\delta(\alpha)d\alpha d\beta,\hskip.4cm
\end{eqnarray}
where the integral in the r.h.s. is to be intended in the sense of Cauchy {\it principal value}. That is
\begin{eqnarray*}
  \int_{S^1\times S^1}^* &=& \lim_{l\rightarrow 0^+}\int_{(S^1\times S^1)\setminus B_l }
\end{eqnarray*}
with $B_l$ a ball or radius $l$ centered at $X(\phi,\theta).$

Explicit expressions of the components of the newtonian force and of their dependence on $R$
will be presented in Section \ref{newtonianforces}. Here we only observe that, due to the axial symmetry of the mass distribution
in ${\mathcal T}_R$ the component on $\epsilon_2$ of the r.h.s. of (\ref{equilibrium}) vanishes
and therefore (\ref{acceleration}) and (\ref{equilibrium}) imply the first integral
\begin{eqnarray}\label{firstintegral}
  (R+r\cos\theta)^2\Omega &=& J(R),
\end{eqnarray}
that expresses the conservation of momentum of momentum with respect to the symmetry axis of ${\mathcal T}_R$.
Equation (\ref{equilibrium}) must be complemented with the continuity equation that
expresses the constance of the flux through the section $S_\theta$ of ${\mathcal T}_R$ obtained
by cutting ${\mathcal T}_R$ at right angle with respect to  ${\mathcal T}_R^0$ along the line $\theta=const:$
\begin{eqnarray}\label{continuity}
  \vert S_\theta\vert v\cdot \epsilon_t  &=& C(R),
\end{eqnarray}
where $\vert S_\theta\vert$ is the measure of $S_\theta$ and $\epsilon_t=n\wedge\epsilon_2$ is a unit vector tangent at $X(\theta,\phi)$ to the line $\phi=const$
 on ${\mathcal T}_R^0$.
Observing that
  \begin{eqnarray}\label{expressions}\left\{\begin{array}{lll}
  \hskip.2cm\epsilon_t&=&{\displaystyle \frac{(r\epsilon_r)_\theta}{\vert(r\epsilon_r)_\theta\vert}=
   \frac{r^\prime\epsilon_r+r\epsilon_\theta}{\sqrt{{r^\prime}^2+r^2}}},\\ \\
     \hskip.2cm n&=&{\displaystyle  \frac{r\epsilon_r-r^\prime\epsilon_\theta}{\sqrt{{r^\prime}^2+r^2}}},\\ \\
     \hskip.2cm\delta&=&{\displaystyle \frac{sr}{\sqrt{{r^\prime}^2+r^2}}},\\ \\
     \vert S_\theta\vert&=&2\pi(R+r\cos\theta)\delta,
   \end{array}
\right.
 \end{eqnarray}
 where $^\prime$ denotes differentiation with respect to $\theta,$
 we can rewrite (\ref{continuity}) in the explicit form
 \begin{eqnarray}\label{continuity1}
   (R+r\cos\theta)sr\omega &=& C(R).
 \end{eqnarray}
By means of the first integral (\ref{firstintegral}) and the continuity equation (\ref{continuity1}) we can determine $s$ and $\Omega$ once $r$ and $\omega$ are known. This allows for transforming system (\ref{equilibrium}),
 (\ref{continuity1}) into an equivalent system, see (\ref{system}) below, for the unknowns $r$ and $\omega$. We let $r_0$ and $\omega_0$  be the averages of $r$ and $\omega$ and we represent the unknowns $r$ and $\omega$ in the form
 \begin{eqnarray}\label{ro-w-form}
 r=r_0(1+\varepsilon\rho),\\\nonumber
\omega= \omega_0(1+\varepsilon w),
 \end{eqnarray}
 where $\varepsilon:=\frac{r_0}{R}$ is regarded as a small parameter and $\rho$ and $w$ are $2\pi-$periodic functions with zero average.
 We observe that the representation of ${\mathcal T}_R$ in (\ref{torus}) is not unique. Indeed
a slight change of $R$ can be exactly compensated by a corresponding change of the function
$\theta\rightarrow r(\theta)$. To make the representation
 (\ref{torus}) of ${\mathcal T}_R$  unique we impose on the unknown $\rho$ the conditions

 \begin{eqnarray}\label{fixcenter}\left\{\begin{array}{lll}
 \int_{S^1}\rho\cos\theta=0,\\
 \int_{S^1}\rho\sin\theta=0.
   \end{array}
\right.
 \end{eqnarray}

Our main result is the following 

\begin{theo}\label{main} Given $r_0>0$ and $\omega_0>0$,
there exists $\varepsilon_0>0$  such that for each $\varepsilon=\frac{r_0}{R}<\varepsilon_0$ system (\ref{system}),
  has a $2\pi-$periodic solution $r=r_0(1+\varepsilon\rho),\,\omega=\omega_0(1+\varepsilon w)$ such that:
  \begin{description}
  \item[(i)] The maps $\rho$ and $w$ are of class $C^{2,\gamma}$ and $C^{1,\gamma}$ respectively for some $\gamma\in(0,1)$. Moreover $\rho$ and $w$ have zero average and $\rho$ satisfies (\ref{fixcenter}).
  \item[(ii)] $\rho$ and $w$ satisfy the estimates
  \begin{eqnarray*}
  \lim_{\varepsilon\rightarrow 0^+}\left\|\rho\right\|_{W^{2,2}}=0,\\
  \lim_{\varepsilon\rightarrow 0^+}\left\|w-\bar{w}\right\|_{W^{1,2}}=0,
  \end{eqnarray*}
  where $\bar{w}:=-\frac{1}{4}\cos\theta$.
  \item[(iii)] The solution is unique in the set of maps that satisfy:
  \begin{eqnarray*}
  \left\|\rho\right\|_{W^{2,2}}+\left\|w\right\|_{W^{1,2}}\leq 2\left\|\bar{w}\right\|_{W^{1,2}}.
  \end{eqnarray*}
  \item[(iv)] The function $s$ is of class $C^{1,\gamma}$ and the function $\Omega$ is of class
  $C^{2,\gamma}$. Moreover
  \begin{eqnarray*}
    \left\|s-\bar{s}\right\|_{W^{1,2}}=o(\varepsilon),\\
  \left\|\Omega-\overline{\Omega}\right\|_{W^{2,2}}=o(\varepsilon),
  \end{eqnarray*}
  where $\bar{s}=\frac{\omega_0^2 r_0}{2\pi\mu G}(1-\varepsilon\frac{3}{4}\cos\theta)$ and
  $\overline{\Omega}=\frac{\omega_0}{2\sqrt{\pi}}\varepsilon\nu(\varepsilon)$ with
  \begin{eqnarray*} \nu(\varepsilon) = O(\sqrt{\log{\frac{1}{\varepsilon}}}).
  \end{eqnarray*}
  \end{description}
 \end{theo}

 Theorem\ref{main}, in the limit case of infinitesimal thickness,
 yields an example of a stationary motion of a self-gravitating fluid which can not
 be reduced to relative equilibrium. It is natural
 to expect that many other mass distribution of a self-gravitating fluid can allow for a similar situation.
 For instance several toroidal strata like ${\mathcal T}_R$ one inside the other with suitable choice of the
 functions $r,s,\omega,\Omega$ for each stratum should do.

  The paper is organized as follows. In sec.2 we reduce system (\ref{equilibrium}),
 (\ref{continuity1}) to a system, see (\ref{system}), of two scalar equation for the unknowns $r,\omega$
 after eliminating $s$ and $\Omega$ via the continuity equation (\ref{continuity1}) and the first integral (\ref{firstintegral}). In Sec.3 we analyze in detail the newtonian forces and study their dependence on the parameter $R.$
 Using the analysis in sec.2 we can write the system for $r,\omega$ as a weak nonlinear equation which,
 for $R>>1$, can be solved leading to the proof of Theorem \ref{main}.

 We denote by $\phi_0,\,\phi_n^j,\,j=1,2,\,n=1,\cdots$ the Fourier coefficient of a $2\pi-$periodic integrable function $\phi$.

\section{The system for $r$ and $\omega$}\label{system-ro-w}
In the following, if $h:\RR\rightarrow\RR$ is a $2\pi$-periodic function, we set
\begin{eqnarray}\label{average}
  \langle h\rangle &=& \int_{S^1}h(\theta)d\theta.
\end{eqnarray}
Let $F$ the integral term on the r.h.s. of (\ref{equilibrium}), that is the force of newtonian interaction. Define
\begin{eqnarray}\label{components}
f^r=F\cdot\epsilon_r,\hskip.3cm f^\theta=F\cdot\epsilon_\theta,\hskip.3cm f=F\cdot\epsilon_1 =f^r\cos\theta-f^\theta\sin\theta.
\end{eqnarray}
The axial symmetry of the problem implies that, as can be also verified by inspecting the expression of $F$,  the components $f^r,\hskip.1cm f^\theta,\hskip.1cm f$ depend only on
the variable $\theta.$ Rewrite
(\ref{equilibrium}) in the form
\begin{eqnarray}\label{equilibrium1}
  v_\theta\omega+v_\phi\Omega&=& F.
\end{eqnarray}
From the kinematic identity
\begin{eqnarray}\label{kinematic}
  v_\theta\cdot\epsilon_1 &=& (v\cdot\epsilon_1)_\theta\hskip.3cm\Rightarrow
  \langle(v\cdot\epsilon_1)_\theta\rangle=0,
\end{eqnarray}
and (\ref{equilibrium1}) it follows
\begin{eqnarray}\label{oomega}
  \langle v_\phi\cdot\epsilon_1\frac{\Omega}{\omega}\rangle=-\langle (R+r\cos(\cdot))\frac{\Omega^2}{\omega}\rangle
  &=& \langle \frac{f}{\omega}\rangle,
\end{eqnarray}
where we have also used (\ref{velocity}) that implies $v_\phi\cdot\epsilon_1=-(R+r\cos\theta)\Omega.$
From (\ref{firstintegral}) and (\ref{oomega}) we obtain
 \begin{eqnarray}\label{jeieomega}\left\{\begin{array}{lll}
 J^2(R)&=&-{\displaystyle\frac{\langle \frac{f}{\omega}\rangle}{\langle\frac{1}{\omega(R+r\cos(\cdot))^3}\rangle}}, \\ \\
   \Omega^2(\theta)&=&-{\displaystyle \frac{1}{(R+r\cos\theta)^4}\frac{\langle \frac{f}{\omega}\rangle}{\langle\frac{1}{\omega(R+r\cos(\cdot))^3}\rangle}}.
   \end{array}
\right.
 \end{eqnarray}
 By means of this expressions of $\Omega^2$ and (\ref{components}), (\ref{acceleration}), from (\ref{equilibrium}),
   after dividing by $\omega$, it follows
   \begin{eqnarray}\label{system}\left\{\begin{array}{lll}
 (r'' - r)\omega  + r' \omega'   &=&{\displaystyle -\frac{\cos(\cdot)}{\omega(R+r\cos(\cdot))^3}\frac{\langle \frac{f}{\omega}\rangle}{\langle\frac{1}{\omega(R+r\cos(\cdot))^3}\rangle}+\frac{f^r}{\omega}}\\ \\
 2r^\prime\omega +r\omega^\prime &=&{\displaystyle\frac{\sin(\cdot)}{\omega(R+r\cos(\cdot))^3}\frac{\langle \frac{f}{\omega}\rangle}{\langle\frac{1}{\omega(R+r\cos(\cdot))^3}\rangle}+\frac{f^\theta}{\omega}}.
   \end{array}
\right.
 \end{eqnarray}

 From (\ref{xthetaxphi}),
 (\ref{expressions}) we obtain
\begin{eqnarray}\label{elementoarea}
  \delta\vert X_\theta\wedge X_\phi\vert &=&\frac{r s}{\sqrt{{r\prime}^2+r^2}}\sqrt{{r\prime}^2+r^2}(R+r\cos(\cdot))\\\nonumber
  &=&r s(R+r\cos(\cdot))=\hskip.1cm C(R)\frac{1}{\omega}.
\end{eqnarray}
We note that from (\ref{elementoarea}), we have
\begin{eqnarray}\label{integralontorus}
  \int_{S^1\times S^1}h \delta\vert X_\theta\wedge X_\phi\vert d\theta d\phi
  &=& C(R)\int_{S^1}\langle\frac{h}{\omega}\rangle d\phi,
\end{eqnarray}
for each function $h:S^1\times S^1\rightarrow\RR$.
 As observed  $f^r, f^\theta, f$  depend only on the variable $\theta$, and therefore
 it suffices to compute these components at $\phi=0$. From this observation and (\ref{elementoarea})
it follows \begin{eqnarray}\label{forcecomponents}\left\{\begin{array}{lll}
  f^r &=& {\displaystyle \frac{G\mu C(R)}{R^2}\int_{S^1\times S^1} \frac{1}{\omega(\alpha)}
  \frac{N^r(\alpha,\beta,\theta)}{D^3(\alpha,\beta,\theta)}d\alpha d\beta},\\ \\
  f^\theta &=& {\displaystyle  \frac{G\mu C(R)}{R^2}\int_{S^1\times S^1}\frac{1}{\omega(\alpha)}
  \frac{N^\theta(\alpha,\beta,\theta)}{D^3(\alpha,\beta,\theta)}d\alpha d\beta},
   \end{array}
\right.
\end{eqnarray}

where
\begin{eqnarray}\label{numerators}\left\{\begin{array}{lll}
 R N^r &=&(X(\alpha,\beta)-X(\theta,0))\cdot\epsilon_r(\theta,0)\\
  &=&-(1-\cos\beta)(R+r(\alpha)\cos\alpha)\cos\theta+r(\alpha)\cos(\alpha-\theta)-r(\theta), \\
 R N^\theta &=&(X(\alpha,\beta)-X(\theta,0))\cdot\epsilon_\theta(\theta,0)\\
  &=&(1-\cos\beta)(R+r(\alpha)\cos\alpha)\sin\theta+r(\alpha)\sin(\alpha-\theta),
  \end{array}
\right.
 \end{eqnarray}
 and
 \begin{eqnarray}\label{denominator}
  R^2 D^2 &=& \vert X(\alpha,\beta)-X(\theta,0)\vert^2\\\nonumber
   &=&((R+r(\alpha)\cos\alpha)\cos\beta-(R+r(\theta)\cos\theta))^2\hskip1.8cm\\\nonumber
   &&+(R+r(\alpha)\cos\alpha)^2\sin^2\beta+(r(\alpha)\sin\alpha-r(\theta)\sin\theta)^2.\hskip1.8cm\\\nonumber
   &=&2(1-\cos\beta)(R^2+R(r(\alpha)\cos\alpha+r(\theta)\cos\theta)+r(\alpha)r(\theta)\cos\alpha\cos\theta)\\\nonumber
   &&+2(1-\cos(\alpha-\theta))r(\alpha)r(\theta)+(r(\alpha)-r(\theta))^2.
 \end{eqnarray}
 In the following we denote by $\frac{a^r}{\omega},\;\frac{a^\theta}{\omega}$ and by ${\mathcal F}^r, \hskip.1cm{\mathcal F}^\theta$ the expressions on l.h.s. and on
 the r.h.s of (\ref{system}) respectively.
 For later reference we list the identities
  \begin{eqnarray}\label{identities}
 \langle\frac{a^r}{\omega}\cos(\cdot)-
 \frac{a^\theta}{\omega}\sin(\cdot)\rangle=0=\langle{\mathcal F}^r\cos(\cdot)-{\mathcal F}^\theta\sin(\cdot)\rangle\hskip.4cm\\
 \langle\frac{a^r}{\omega}\sin(\cdot)+
 \frac{a^\theta}{\omega}\cos(\cdot)\rangle=0=\langle{\mathcal F}^r\sin(\cdot)+{\mathcal F}^\theta\cos(\cdot)\rangle.\hskip.4cm
 \end{eqnarray}
 The first two are just a rewriting of (\ref{kinematic}) and (\ref{oomega}), the third is equivalent
 to $\langle((r\sin(\cdot))^\prime\omega)^\prime\rangle=0$ and says that the component of the
 acceleration of the center of mass of ${\mathcal T}_R$ on the axis of ${\mathcal T}_R$ is zero. The fourth
 identity is equivalent to $\langle\frac{F\cdot e_3}{\omega}\rangle=0$ which is a consequence
 of the fact that there are no exterior forces acting on ${\mathcal T}_R$. In the following we will also use the identities
 \begin{eqnarray}\label{newidentity}
 \langle\omega (r^\prime{\mathcal F}^r+r {\mathcal F}^\theta)\rangle=0,\\\nonumber
 \langle r^\prime a^r+r a^\theta \rangle=0.
 \end{eqnarray}
The identity (\ref{newidentity})$_1$ says that the power $p$ of the conservative field of force $\Phi$ sum of the newtonian and centrifugal forces on the stationary motion of the torus is zero. To derive (\ref{newidentity})$_1$ we note that (\ref{integralontorus}) implies:
\begin{eqnarray}
0=p=\int_{S^1\times S^1}\Phi\cdot v\delta\vert X_\theta\wedge X_\phi| d\theta d\phi=C(R)\int_{S^1}\langle\frac{\Phi\cdot v}{\omega}\rangle d\phi\\\nonumber
=2\pi C(R)\langle\omega (r^\prime{\mathcal F}^r+r {\mathcal F}^\theta)\rangle,\hskip2cm
\end{eqnarray}
where we have also used the definition of ${\mathcal F}^r,\;{\mathcal F}^\theta$ and (\ref{velocity}) that implies $\frac{\Phi\cdot v}{\omega}=\omega(r^\prime{\mathcal F}^r+r {\mathcal F}^\theta)$.
The kinematic identity (\ref{newidentity})$_2$ follows from
\begin{eqnarray}
\langle r^\prime a^r+r a^\theta \rangle &=&\langle\omega[r^\prime((r^{\prime\prime}-r)\omega+(r^2\omega)^\prime]\rangle\\\nonumber
&=&\langle\omega[({r^\prime}^2\omega)^\prime-\frac{\omega}{2}(r^2+{r^\prime}^2)^\prime+(r^2\omega)^\prime]\rangle\\\nonumber
&=&\langle-({r^\prime}^2\omega)\omega^\prime+\frac{(\omega^2)^\prime}{2}(r^2+{r^\prime}^2)-(r^2\omega)\omega^\prime\rangle\\\nonumber &=&0.
\end{eqnarray}
 We can replace (\ref{system}) with the equivalent system
\begin{eqnarray}\label{newsystem}
\left\{\begin{array}{l}\frac{a^r}{\omega}={\mathcal F}^r,\\
\frac{a^\theta}{\omega}={\mathcal F}^\theta+\frac{r^\prime}{r}({\mathcal F}^r-\frac{a^r}{\omega})  .
\end{array}\right.
\end{eqnarray}
The advantage of (\ref{newsystem}) with respect to the original system (\ref{system}) is that
 (\ref{newidentity}) imply that, in the analysis of (\ref{newsystem}), we don't need to consider the projection of (\ref{newsystem})$_2$ on the subspace of constants functions.
 This is a consequence of the following lemma for $f= \omega (r^\prime{\mathcal F}^r+r {\mathcal F}^\theta)-(r^\prime a^r+r a^\theta)$ and $h=\frac{1}{r\omega}$.
 \begin{lem}\label{can-be-done}
 Let $f,\,h$ be $2\pi-$periodic $L^2(-\pi,\pi)$ functions. Assume that $c\leq h\leq C$ for some constants $c,\,C>0$. Let $g=h f,\,k=\frac{1}{h}$. Then a necessary a sufficient condition in order that $g=0$ is that $f_0=0$ and $g_n^j=0,\,n=1,\cdots;\,j=1,2$.
 \end{lem}
\begin{proof}
A standard computation reveals that
\begin{eqnarray}
f_0=k_0 g_0+\sum_{n\geq 1}(k_n^i g_n^1+k_n^2 g_n^2).
\end{eqnarray}
The lemma follows from this and from the assumptions on $h$ that imply $k_0>0$.
\end{proof}
 \section{The Newtonian Forces}
\label{newtonianforces}

In this section we assume that the unknowns $\rho$ and $w$ in (\ref{ro-w-form}) are $2\pi-$periodic functions such that
 \begin{eqnarray}\label{smoothness}
   \rho\in C^{1,\gamma}, & & w\in C^{0,\gamma},\;\,\gamma\in(0,1)\\\nonumber
   \vert|\rho\vert|_{C^{1,\gamma}}\leq C, & & \vert|w\vert|_ {C^{0,\gamma}}\leq C.
 \end{eqnarray}
 We analyze the smoothness and the dependence on $\varepsilon=\frac{r_0}{R}$ of the forcing terms ${\mathcal F}^r,\hskip.1cm{\mathcal F}^\theta$ on the r.h.s. of (\ref{system}). This analysis involves the study of certain integral operators with singular kernels which, even though are of the type considered in the literature \cite{miranda}, are extended to a manifold which depends on the singular parameter $\varepsilon$. For this reason,
  in order to estimate the dependence on $\varepsilon<<1$ of the norms of these operators
 we develop a direct analysis.

 By using
 (\ref{ro-w-form}), after setting
 \begin{eqnarray}\label{zeta}
 z(t) &=& 2(1-\cos t),
 \end{eqnarray}
 we rewrite (\ref{numerators}) as
 \begin{eqnarray}\label{numerators1}\left\{\begin{array}{lll}
   N^r &=& \frac{1}{2}(-z(\beta)\cos\theta +\varepsilon(-z(\beta)\cos\alpha\cos\theta-z(\alpha-\theta))\\
   & &+\varepsilon^2(-z(\beta)\rho(\alpha)\cos\alpha\cos\theta
   -z(\alpha-\theta)\rho(\alpha)+2(\rho(\alpha)-\rho(\theta)))),\\
   N^\theta &=& \frac{1}{2}(z(\beta)\sin\theta+\varepsilon(z(\beta)\cos\alpha\sin\theta
   +2\sin(\alpha-\theta))\\
   &&+\varepsilon^2(z(\beta)\rho(\alpha)\cos\alpha\sin\theta+2\rho(\alpha)\sin(\alpha-\theta))).
   \end{array}
\right.
 \end{eqnarray}
 A similar computation leads to
 \begin{eqnarray}\label{denominator1}
   D^2 &=& z(\beta)[1+\varepsilon(\cos\alpha+\cos\theta)\\\nonumber
   &&+\varepsilon^2(\cos\alpha\cos\theta+\rho(\alpha)\cos\alpha+\rho(\theta)\cos\theta)\\\nonumber
   &&+\varepsilon^3(\rho(\alpha)+\rho(\theta))\cos\alpha\cos\theta
   +\varepsilon^4\rho(\alpha)\rho(\theta)\cos\alpha\cos\theta]\\\nonumber
   &&+\varepsilon^2z(\alpha-\theta)[1+\varepsilon(\rho(\alpha)+\rho(\theta))+\varepsilon^2(\rho(\alpha)\rho(\theta)\\\nonumber
   &&+\frac{(\rho(\alpha)-\rho(\theta))^2}{z(\alpha-\theta)})] .
 \end{eqnarray}

 \begin{lem}\label{integrands}
  The components $f^r$ and $f^\theta$ given by (\ref{forcecomponents}) can be expressed in the form:
  \begin{eqnarray}\label{componentsexpansion}\left\{\begin{array}{lll}
 f^r=\frac{G\mu C(R)}{2\omega_0 R^2}\sum_{h=1}^3\int_{S^1\times S^1}
 K_h^r(\beta,\alpha-\theta,\varepsilon)(1+S_h^r(\rho,w;\alpha,\beta,\theta,\varepsilon))d\alpha d\beta,\\
 f^\theta=\frac{G\mu C(R)}{2\omega_0 R^2}\sum_{h=1}^2\int_{S^1\times S^1}
 K_h^\theta(\beta,\alpha-\theta,\varepsilon)(1+S_h^\theta(\rho,w;\alpha,\beta,\theta,\varepsilon))d\alpha  d\beta,
   \end{array}
\right.
 \end{eqnarray}
 where 
 \begin{eqnarray}\label{kernels}
   K_1^r &=& \frac{2\varepsilon^2(\rho(\alpha)-\rho(\theta))}{(z(\beta)+\varepsilon^2z(\alpha-\theta))^\frac{3}{2}}, \\\nonumber
   K_2^r=\frac{1}{\varepsilon}K_2 &=& -\frac{\varepsilon z(\alpha-\theta)}{(z(\beta)+\varepsilon^2z(\alpha-\theta))^\frac{3}{2}}, \\\nonumber
   K_3^r=-\cos(\theta)K_3 &=& -\frac{z(\beta)\cos\theta}{(z(\beta)+\varepsilon^2z(\alpha-\theta))^\frac{3}{2}}, \\\nonumber
   K_1^\theta=\frac{1}{\varepsilon}K_1 &=& \frac{2\varepsilon \sin(\alpha-\theta)}{(z(\beta)+\varepsilon^2z(\alpha-\theta))^\frac{3}{2}},  \\\nonumber
   K_2^\theta=\sin(\theta)K_3 &=& \frac{z(\beta)\sin\theta}{(z(\beta)+\varepsilon^2z(\alpha-\theta))^\frac{3}{2}},
 \end{eqnarray}

 and, under the assumption (\ref{smoothness}), $S_h^r(\rho,w)$ and $S_h^\theta(\rho,w)$  are $C^{0,\gamma}$ functions such that

 \begin{eqnarray}\label{ssmoothness}
   \vert| S_h^r(\rho,w)\vert|_{C^{0,\gamma}},\vert| S_h^\theta(\rho,w)\vert|_{C^{0,\gamma}}&\leq& C\varepsilon ,\\\nonumber
   \vert| S_h^r(\rho_1,w_1)-S_h^r(\rho_2,w_2)\vert|_{C^{0,\gamma}}&\leq&
   C\varepsilon\vert|(\rho_1,w_1)-(\rho_2,w_2)\vert|,\\\nonumber
   \vert| S_h^\theta(\rho_1,w_1)-S_h^\theta(\rho_2,w_2)\vert|_{C^{0,\gamma}}&\leq&
   C\varepsilon\vert|(\rho_1,w_1)-(\rho_2,w_2)\vert|,
 \end{eqnarray}
 where $\vert|(\rho,w)\vert|:=\vert|\rho\vert|_{C^{1,\gamma}}+\vert|w\vert|_{C^{0,\gamma}}.$

 \end{lem}
 Proof. If we factor out $D_0^2(\beta,\alpha-\theta)=z(\beta)+\varepsilon^2z(\alpha-\theta)$ from the expression of
    $D^2$ given by (\ref{denominator1}) we get  $D^2=D_0^2(1+d(\rho;\alpha,\beta,\theta,\varepsilon))$ with
    \begin{eqnarray}\label{denominator2}
  d(\rho;\alpha,\beta,\theta,\varepsilon) &=& \frac{z(\beta)}{D_0^2(\beta,\alpha-\theta)} (\varepsilon(\cos\alpha+\cos\theta )\\\nonumber
   &&+\varepsilon^2(\cos\alpha\cos\theta+\rho(\alpha)\cos\alpha+\rho(\theta)\cos\theta)\\\nonumber
   &&+\varepsilon^3(\rho(\alpha)+\rho(\theta))\cos\alpha\cos\theta
   +\varepsilon^4\rho(\alpha)\rho(\theta)\cos\alpha\cos\theta)\\\nonumber
   &&+\frac{\varepsilon^2 z(\alpha-\theta)}{D_0^2(\beta,\alpha-\theta)}(\varepsilon(\rho(\alpha)+\rho(\theta))+\epsilon^2(\rho(\alpha)\rho(\theta)\\\nonumber
   &&+\frac{(\rho(\alpha)-\rho(\theta))^2}{z(\alpha-\theta)})) .
 \end{eqnarray}

    From (\ref{smoothness}) and (\ref{denominator2}) it follows
    \begin{eqnarray}\label{drhosmoothness}
      \|d(\rho)\|_{C^{0,\gamma}} &\leq& C\varepsilon(1+\vert|\rho \vert|_{C^{1,\gamma}}),\\\nonumber
      \|d(\rho_1)-d(\rho_2)\|_{C^{0,\gamma}} &\leq& C\varepsilon\vert|\rho_1-\rho_2 \vert|_{C^{1,\gamma}}.
    \end{eqnarray}
    and therefore we have

  \begin{eqnarray}\label{unosud}\frac{1}{D^3}=\frac{1}{D_0^3(\beta,\alpha-\theta)}\left(1+\sum_{n=1}^\infty \left(\begin{array}{c}
 -\frac{3}{2}\\
 n
   \end{array}
\right)d(\rho;\alpha,\beta,\theta,\varepsilon)^n\right),
 \end{eqnarray}
where, if $\varepsilon\in(0,\varepsilon_0]$ for some $\varepsilon_o>0$, the series on the r.h.s converges
absolutely and uniformly on $S^1\times S^1\times S^1$.

 Similarly
   \begin{eqnarray}\label{unosuomega}\frac{1}{\omega}=\frac{1}{\omega_0}\left(1+\sum_{m=1}^\infty \varepsilon^m\left(\begin{array}{c}
 -1\\
 m
   \end{array}
\right)w^m\right)=\frac{1}{\omega_0}\left(1+\sum_{m=1}^\infty(-1)^m\varepsilon^mw^m\right).
 \end{eqnarray}
 In conclusion
 we can rewrite (\ref{forcecomponents}) in the form

  \begin{eqnarray}\label{forcecomponentfr}
f^r &=&\frac{G\mu C(R)}{2\omega_0 R^2}\int_{S^1\times S^1}\frac{2N^r}{D_0^3}
\left(1+S(\rho,w)\right)d\alpha d\beta,\\\nonumber
 f ^\theta &=&\frac{G\mu C(R)}{2\omega_0 R^2}\int_{S^1\times S^1}\frac{2N^\theta}{D_0^3} \left(1+S(\rho,w)\right)d\alpha d\beta.
 \end{eqnarray}

 where
 \begin{eqnarray}\label{esse}
   S(\rho,w)&=& \sum_{n=1}^\infty \sum_{m=0}^n(-1)^m\varepsilon^m \left(\begin{array}{c}
 -\frac{3}{2}\\
 n-m
   \end{array}
\right)w^m d(\rho)^{n-m}
   \end{eqnarray}

 From (\ref{drhosmoothness}) and the inequalities
 \begin{eqnarray}
 \vert|\Pi_1^nu_j\vert|_{C^{0,\gamma}} &\leq& (n+1)\Pi_1^n\vert|u_j\vert|_{C^{0,\gamma}}  \\\nonumber
 \vert|\Pi_1^nu_j-\Pi_1^nv_j\vert|_{C^{0,\gamma}} &\leq& \\\nonumber &&n\sum_{h=1}^n\Pi_0^{h-1}\vert|v_j\vert|_{C^{0,\gamma}}\Pi_{h+1}^n\vert|u_j\vert|_{C^{0,\gamma}}
 \vert|u_h-v_h\vert|_{C^{0,\gamma}}
 \end{eqnarray}
 valid for $u_j,v_j\in C^{0,\gamma},j=1,\dots,n,\;\; u_0=u_{n+1}=v_0=v_{n+1}=1$  we see that
 \begin{eqnarray}
   \vert|S(\rho,w)\vert|_{C^{0,\gamma}}&\leq& \sum_{n=1}^\infty \varepsilon^n(n+1)C^n\sum_{m=0}^n \left\vert\left(\begin{array}{c}
 -\frac{3}{2}\\
 n-m
   \end{array}
\right)\right\vert \\\nonumber
 &\leq&\sum_{n=1}^\infty \varepsilon^nC_1^n\\\nonumber
 \vert|S(\rho_1,w_1)-S(\rho_2,w_2)\vert|_{C^{0,\gamma}}&\leq&\sum_{n=1}^\infty \varepsilon^nn^2C^{n-1}\sum_{m=0}^n \left\vert\left(\begin{array}{c}
 -\frac{3}{2}\\
 n-m
   \end{array}
\right)\right\vert \vert|(\rho_1,w_1)-(\rho_2,w_2)\vert|\\\nonumber
 &\leq&\sum_{n=1}^\infty \varepsilon^n C_1^n\vert|(\rho_1,w_1)-(\rho_2,w_2)\vert|,
 \end{eqnarray}
 where $C \text{ and } C_1$ are constants independent on $n$.
  From this and  the expressions
  (\ref{numerators1}) of $N^r$ and $N^\theta$ that we rewrite as
  \begin{eqnarray}\label{numerators2}\left\{\begin{array}{lll}
   2N^r &=& -z(\beta)\cos\theta(1 +\varepsilon \cos\alpha
   +\varepsilon^2 \rho(\alpha)\cos\alpha)\\ & &
   -\varepsilon z(\alpha-\theta)(1+\varepsilon\rho(\alpha))+2\varepsilon^2(\rho(\alpha)-\rho(\theta)),\\
   2N^\theta &=&z(\beta)\sin\theta(1+\varepsilon\cos\alpha+\varepsilon^2\rho(\alpha)\cos\alpha)\\& &  +2\varepsilon\sin(\alpha-\theta)(1+\varepsilon\rho(\alpha)),
   \end{array}
\right.
 \end{eqnarray}  the lemma follows.

   \begin{lem}\label{mapdefinedbyk}
 Let $\sigma:S^1\times S^1\times S^1\rightarrow\RR$ a map of class $C^{0,\gamma}, \gamma\in(0,1).$
 Then
 \begin{eqnarray}\label{int-op}
  \mathcal{K}_1^r\sigma(\theta) &=& \int_{S^1\times S^1}^*K_1^r(\alpha,\beta,\theta,\varepsilon)\sigma(\alpha,\beta,\theta)d\alpha d\beta,\hskip.3cm  \\
   \mathcal{K}_j \sigma(\theta) &=& \int_{S^1\times S^1}^*K_j (\alpha,\beta,\theta,\varepsilon)\sigma(\alpha,\beta,\theta)d\alpha d\beta,\hskip.3cm j=1,2,3,
   \end{eqnarray}
 define  continuous linear maps $\mathcal{K}_1^r, \mathcal{K}_j :C^{0,\gamma}\rightarrow C^{0,\gamma^\prime}, j=1,2,3$ for
 all $0<\gamma^\prime<\frac{\gamma}{1+\gamma}.$ Moreover:
 \begin{eqnarray}\label{bounds}
   \vert|\mathcal{K}_1^r\vert| &<& C_{\gamma^\prime} \vert|\rho\vert|_{C^{1,\gamma}} ,\\\nonumber
   \vert|\mathcal{K}_1 \vert| &<& C_{\gamma^\prime},\\\nonumber
   \vert|\mathcal{K}_2 \vert|  &<& C ,\\\nonumber
   \vert|\mathcal{K}_3 \vert| &<& C(1+\log{\frac{1}{\varepsilon}}),
 \end{eqnarray}
where $C$ and $C_{\gamma^\prime}$ are constants and $\lim_{\gamma^\prime\rightarrow\frac{\gamma}{1+\gamma}} C_{\gamma^\prime}=+\infty.$
 \end{lem}
\begin{proof} If we replace $\alpha$ with $\alpha+\theta$ in (\ref{int-op})$_1$, we have

  \begin{eqnarray}\label{restrictedintegral}
    \mathcal{K}_1^r\sigma(\theta) &=& \int_{(-\pi,\pi)^2}^*
    K_1^r(\alpha+\theta,\beta,\theta,\varepsilon)\sigma(\alpha+\theta,\beta,\theta)d\alpha d\beta.
  \end{eqnarray}
Given $\psi:(S^1)^3\rightarrow\RR$ let
\begin{eqnarray}\label{evenodd}
    \psi^\pm(\alpha+\theta,\beta,\theta)&=&\frac{1}{2}(\psi(\alpha+\theta,\beta,\theta)
    \pm\psi(-\alpha+\theta,\beta,\theta)),
    \end{eqnarray}
and observe that $\psi\in C^{0,\gamma}\Rightarrow \vert|\psi^\pm\vert|_{C^{0,\gamma}}\leq\vert|\psi \vert|_{C^{0,\gamma}}$ and moreover that:
\begin{description}
  \item[(i)] $\psi\in C^{k,\gamma},\hskip.2cm k=0,1\hskip.4cm\Rightarrow $
   $$\vert\psi^-(\alpha+\theta,\beta,\theta)
     \vert\leq\vert|\psi\vert|_{C^{0,\gamma}}
    \vert\alpha\vert^\gamma,\hskip.2cm\rm{if}\hskip.2cm k=0,$$
   $$\vert\psi^-(\alpha+\theta,\beta,\theta)
     \vert\leq\vert|\psi\vert|_{C^{1,\gamma}}
    \vert\alpha\vert,\hskip.2cm\rm{if}\hskip.2cm k=1,$$
     \item[(ii)] $\psi\in C^{1,\gamma},  \hskip.4cm \rm{and}\hskip.2cm \psi(\theta,\beta,\theta)=0\hskip.4cm  \Rightarrow $
  $$\vert\psi^+(\alpha+\theta,\beta,\theta)
     \vert\leq\vert|\psi\vert|_{C^{1,\gamma}}
    \vert\alpha\vert^{1+\gamma},$$
      \item[(iii)] $\psi\in C^{0,\gamma}\hskip.4cm\Rightarrow $
   $$\vert\psi^-(\alpha+\theta_1,\beta,\theta_1)
    -\psi^-(\alpha+\theta_2,\beta,\theta_2)\vert\leq\vert|\psi\vert|_{C^{0,\gamma}}
    \vert\alpha\vert^\nu\vert\theta_1-\theta_2\vert^{\gamma-\nu}.$$
    If $ \psi(\theta,\beta,\theta)=0,$ the same is true for $\psi^+.$ Otherwise it results
     $$\vert\psi^+(\alpha+\theta_1,\beta,\theta_1)
    -\psi^+(\alpha+\theta_2,\beta,\theta_2)\vert\leq\vert|\psi\vert|_{C^{0,\gamma}}
     \vert\theta_1-\theta_2\vert^\gamma\hskip1.2cm$$
     \item[(iv)] $\psi\in C^{1,\gamma}   \hskip.4cm   \Rightarrow $
  $$\vert\psi^-(\alpha+\theta_1,\beta,\theta_1)
    -\psi^-(\alpha+\theta_2,\beta,\theta_2)\vert\leq\vert|\psi\vert|_{C^{1,\gamma}}
    \vert\alpha\vert^{1-\nu^\prime}\vert\theta_1-\theta_2\vert^{\nu^\prime},$$
  \item[(v)] $\psi\in C^{1,\gamma},   \hskip.4cm \rm{and}\hskip.2cm \psi(\theta,\beta,\theta)=0\hskip.4cm  \Rightarrow $
  $$\vert\psi^+(\alpha+\theta_1,\beta,\theta_1)
    -\psi^+(\alpha+\theta_2,\beta,\theta_2)\vert\leq\vert|\psi\vert|_{C^{1,\gamma}}
    \vert\alpha\vert^{(1+\gamma)(1-\nu^\prime)}\vert\theta_1-\theta_2\vert^{\nu^\prime},$$
   \end{description}
     where $\nu\in[0,\gamma]$ and $\nu^\prime\in[0,1]$ are arbitrary numbers.   Let $\mathcal{K}_1^{r,\pm}$ the operator defined by (\ref{restrictedintegral}) when $K_1^r$ is replaced by $K_1^{r,\pm}.$ Then we have:
     \begin{eqnarray*}
       \mathcal{K}_1^{r,\pm}\sigma^\mp &=& 0,
     \end{eqnarray*}
     and therefore:
      \begin{eqnarray*}
       \mathcal{K}_1^r\sigma &=& \mathcal{K}_1^{r,+}\sigma^++ \mathcal{K}_1^{r,-}\sigma^-.
     \end{eqnarray*}
     Applying {\rm (ii)} to the function $\rho(\alpha)-\rho(\theta)$ we get after setting $\beta=\varepsilon\eta$
      \begin{eqnarray}\label{evenpart}
       \vert\mathcal{K}_1^{r,+}\sigma^+\vert&\leq& \vert|\rho\vert|_{C^{1,\gamma}}\vert|\sigma\vert|_{C^{0,\gamma}}
       \int_{-\pi}^{\pi}\int_{-\pi}^{\pi}d\alpha d\eta
       \frac{\vert\alpha\vert^{1+\gamma}}{(\eta^2\frac{z(\varepsilon\eta)}{\varepsilon^2\eta^2}+z(\alpha))^\frac{3}{2}}\\\nonumber &\leq& Const\vert|\rho\vert|_{C^{1,\gamma}}\vert|\sigma\vert|_{C^{0,\gamma}},
     \end{eqnarray}
     where we have also observed that
 \begin{eqnarray}\label{zeta-su-eta}
 1>\frac{z(s)}{s^2}\geq\frac{4}{\pi^2},\;\;s\in(0,\pi].
 \end{eqnarray}
     Applying {\rm (i)} with $k=1$ to $\rho(\alpha)-\rho(\theta)$ and {\rm (i)} with $k=0$ to $\sigma$ we see that the estimate (\ref{evenpart}) is valid for $\vert\mathcal{K}_1^{r,-}\sigma^-\vert$ too and we conclude that
       \begin{eqnarray}\label{cestimate}
       \vert|\mathcal{K}_1^r \sigma\vert|_{C^0}&\leq& Const\vert|\rho\vert|_{C^{1,\gamma}}\vert|\sigma\vert|_{C^{0,\gamma}}
       .
  \end{eqnarray}
  We also have
  \begin{eqnarray}\label{evenpartholder}
      \vert\mathcal{K}_1^{r,+}\sigma^+(\theta_1)-\mathcal{K}_1^{r,+}\sigma^+(\theta_2)\vert\leq \hskip3cm  \\\nonumber
       \int_{(-\pi,\pi)^2}
       \vert K _1^{r,+}(\theta_1)-K _1^{r,+}(\theta_2)\vert\vert\sigma^+\vert d\alpha d\beta +\int_{(-\pi,\pi)^2} \vert K _1^{r,+}\vert\vert\sigma^+(\theta_1)-\sigma^+(\theta_2)\vert d\alpha d\beta.
     \end{eqnarray}
     To estimate the first integral in (\ref{evenpartholder}) we apply {\rm (v)} with $\nu^\prime=\gamma^\prime<\frac{\gamma}{1+\gamma}$ to $\rho(\alpha)-\rho(\theta)$,   for the second integral we use instead   {\rm (ii)}   for $\rho(\alpha)-\rho(\theta)$ and {\rm (iii)}   for $\sigma.$
    With the change of variable $\beta=\varepsilon\eta$ this yields
     \begin{eqnarray}\label{evenpartholder1}
       \vert\mathcal{K}_1^{r,+}\sigma^+(\theta_1)-\mathcal{K}_1^{r,+}\sigma^+(\theta_2)\vert\leq \vert|\rho\vert|_{C^{1,\gamma}}\vert|\sigma\vert|_{C^{0,\gamma}}\vert\theta_1-\theta_2\vert^{\gamma^\prime} &&\\\nonumber
       \int_{-\pi}^{\pi}\int_{-\frac{\pi}{\varepsilon}}^{\frac{\pi}{\varepsilon}}
       \frac{\vert\alpha\vert^{(1+\gamma)(1-\gamma^\prime)}+\vert\alpha\vert^{1+\gamma  }}{(\eta^2\frac{z(\varepsilon\eta)}{\varepsilon^2\eta^2}+z(\alpha))^\frac{3}{2}}d\alpha d\eta.
      \end{eqnarray}It follows
     \begin{eqnarray}\label{cgammaestimate}
       \frac{\vert\mathcal{K}_1^{r,+}\sigma^+(\theta_1)-\mathcal{K}_1^{r,+}\sigma^+(\theta_2)\vert}
       {\vert\theta_1-\theta_2\vert^{\gamma^\prime}} &\leq&  C_{\gamma^\prime}\vert|\rho\vert|_{C^{1,\gamma}}\vert|\sigma\vert|_{C^{0,\gamma}},
     \end{eqnarray} where $C_{\gamma^\prime}$ depends on $0<\gamma^\prime<\frac{\gamma}{1+\gamma}.$
     This and (\ref{cestimate})
     imply
     \begin{eqnarray*}
       \vert|\mathcal{K}_1^{r,+}\sigma^+\vert|_{C^{0,\gamma^\prime}} &\leq&  C_{\gamma^\prime}\vert|\rho\vert|_{C^{1,\gamma}}\vert|\sigma\vert|_{C^{0,\gamma}}.
     \end{eqnarray*} In a similar way we establish that also $K_1^{r,-}\sigma^-$ satisfies (\ref{cgammaestimate}).
     Indeed, after writing (\ref{evenpartholder}) for $K_1^{r,-}\sigma^-$, we use, for the first integral, {\rm (iv)}   with $\nu^\prime=\gamma^\prime$ for $\rho(\alpha)-\rho(\theta)$ and {\rm (i)}  with $k=1$ for $\sigma$ and for the second integral {\rm (i)} with $k=0,$  for $\rho(\alpha)-\rho(\theta)$ and {\rm (iii)} with $\nu=\gamma-\gamma^\prime$ for $\sigma.$ This concludes the proof for the operator $\mathcal{K}_1^r$.
     The analysis of the other operators  follows the same path but it is simpler due to the fact that the corresponding kernels   depend on $\alpha$ and $\theta$ only through the difference $\alpha- \theta$. Moreover the singularity of the kernel $K_1 $ is weaker while the kernels $K_2$ and $K_3$ are not singular. To estimate the norms of  $\mathcal{K}_3$ we use again the change of variable $\beta=\varepsilon\eta$. Then  we have
  \begin{eqnarray}\label{evenpart-1}
      \hskip1cm\mathcal{K}_3
       &=&4\int_0^{\pi}\int_0^{\frac{\pi}{\varepsilon}}d\alpha d\eta
       \frac{\eta^2\frac{z(\varepsilon\eta)}{\varepsilon^2\eta^2}}
       {(\eta^2\frac{z(\varepsilon\eta)}{\varepsilon^2\eta^2}+z(\alpha))^\frac{3}{2}}\\\nonumber
       &=&4\int_0^{\pi}\int_1^{\frac{\pi}{\varepsilon}}d\alpha d\eta
       \frac{\eta^2\frac{z(\varepsilon\eta)}{\varepsilon^2\eta^2}}
       {(\eta^2\frac{z(\varepsilon\eta)}{\varepsilon^2\eta^2}+z(\alpha))^\frac{3}{2}}+C,
       \end{eqnarray}
       and therefore using also (\ref{zeta-su-eta})
       \begin{eqnarray}\label{k3-two-pieces}
       \mathcal{K}_3
       \leq 2\pi^2\int_1^{\frac{\pi}{\varepsilon}}\frac{d\eta}{\eta}+C\leq
       C_1(1+\log{\frac{1}{\varepsilon}}).
     \end{eqnarray}
     This completes the proof.
     \end{proof}
\begin{rem}\label{k3-remark}
We also have
\begin{eqnarray}\label{k3-from-below}
\mathcal{K}_3\geq c_1(1+\log{\frac{1}{\varepsilon}}).
\end{eqnarray}
This follows from (\ref{evenpart-1}) (\ref{zeta-su-eta}) and the fact that $x\rightarrow\frac{x^3}{(x^2+a)^\frac{3}{2}},\, a>0,$ is increasing for $x\in(0,+\infty)$ and therefore
\begin{eqnarray}
\int_0^{\pi}\int_1^\frac{\pi}{\varepsilon}d\alpha d\eta
       \frac{\eta^2\frac{z(\varepsilon\eta)}{\varepsilon^2\eta^2}}
       {(\eta^2\frac{z(\varepsilon\eta)}{\varepsilon^2\eta^2}+z(\alpha))^\frac{3}{2}}\geq
  \int_0^{\pi}d\alpha\frac{(\frac{4}{\pi^2})^\frac{3}{2}}
  {(\frac{4}{\pi^2}+z(\alpha))^\frac{3}{2}}\int_1^\frac{\pi}{\varepsilon}\frac{d\eta}{\eta},
\end{eqnarray}
which implies (\ref{k3-from-below}).
\end{rem}
In the following the symbol $O(\varepsilon^k)^{0,\gamma^\prime}$ stands for a map $h(\rho,w;\alpha,\theta,\varepsilon)$ which, provided assumption (\ref{smoothness}) holds, is such that

 \begin{eqnarray}\label{ssmoothness-h}
   \vert| h(\rho,w)\vert|_{C^{0,\gamma^\prime}} &\leq& C\varepsilon^k ,\\\nonumber
   \vert| h(\rho_1,w_1)-h(\rho_2,w_2)\vert|_{C^{0,\gamma^\prime}}&\leq&
   C\varepsilon^k\vert|(\rho_1,w_1)-(\rho_2,w_2)\vert|,
 \end{eqnarray}
 for some constant $C>0$ independent of $\varepsilon$.

 If instead (\ref{ssmoothness-h}) holds with a constant $C=C(\varepsilon)$ that   depends on $\varepsilon$ and $\lim_{\varepsilon\rightarrow 0^+} C(\varepsilon)=0$ we say that $h(\rho,w;\alpha,\theta,\varepsilon)=o(\varepsilon^k)^{0,\gamma^\prime}$.
 A key step in the proof of Theorem \ref{main} is the characterization and the analysis of the terms of order $1$ and $\varepsilon$ of the expressions ${\mathcal F}^r, \hskip.1cm{\mathcal F}^\theta$   on
 the r.h.s of (\ref{system}).

 \begin{lem}\label{linearpart} Set $c(\varepsilon)=\frac{G\mu C(R)}{2\omega_0^3r_0^2 R}$ then
 \begin{eqnarray}\label{rhslinearpart}
   \frac{{\mathcal F}^r}{\omega_0r_0} &=& \varepsilon c(\varepsilon)(
    \frac{1}{\varepsilon}\mathcal{K}_2 +(\mathcal{K}_1^r+ \mathcal{K}_2\sigma-w\mathcal{K}_2)\\\nonumber
   & &-\frac{\cos(\cdot)}{2\pi}\langle\cos(\cdot)
   (
    \mathcal{K}_1^r+\mathcal{K}_2\sigma-w\mathcal{K}_2 )-\sin(\cdot)\mathcal{K}_1\sigma
  \rangle) +o(\varepsilon)^{0,\gamma^\prime},
  \\\nonumber
   \frac{{\mathcal F}^\theta}{\omega_0r_0} &=&\varepsilon c(\varepsilon)
    (\mathcal{K}_1\sigma  \\\nonumber
   & &+\frac{\sin(\cdot)}{2\pi}\langle\cos(\cdot)
   (
    \mathcal{K}_1^r+\mathcal{K}_2\sigma-w\mathcal{K}_2 )-\sin(\cdot)\mathcal{K}_1\sigma
  \rangle) +o(\varepsilon)^{0,\gamma^\prime},
  \end{eqnarray}
 where
 \begin{eqnarray}\label{sigma}
        \sigma &=& \rho(\alpha)-w(\alpha)\\\nonumber&&-\frac{3}{2}
        (\frac{z(\beta)}{D_0^2(\beta,\alpha-\theta)}(\cos\alpha+\cos\theta)
        +\frac{\varepsilon^2 z(\alpha-\theta)}{D_0^2(\beta,\alpha-\theta)}(\rho(\alpha)+\rho(\theta)).
      \end{eqnarray}
 \end{lem}
 Proof. From (\ref{esse}) it follows $S(\rho,w)=-\varepsilon w+\frac{3}{2}d(\rho)+o(\varepsilon)^{0,\gamma^\prime}$. This, (\ref{forcecomponentfr}), (\ref{numerators2}) and (\ref{unosuomega}) imply
   \begin{eqnarray}\label{forcecomponentorder1}
\hskip1.5cm\frac{f^r}{\omega} &=&\frac{G\mu C(R)}{2\omega_0^2 R^2}(1-\varepsilon w)(\mathcal{K}_1^r+\mathcal{K}_3^r+\mathcal{K}_2^r+\varepsilon\mathcal{K}_2^r\rho)\\\nonumber
&&\left(1-\varepsilon w-\frac{3}{2}d(\rho) \right)  +o(\varepsilon)^{0,\gamma^\prime},\\\nonumber
&=&\omega_0r_0\varepsilon c(\varepsilon)(\mathcal{K}_1^r -\cos\theta\mathcal{K}_3+\frac{1}{\varepsilon}\mathcal{K}_2
+\mathcal{K}_2(\rho-w-\frac{3}{2\varepsilon}d(\rho))-w(\theta)\mathcal{K}_2)+o(\varepsilon)^{0,\gamma^\prime},\\\nonumber
&=&\omega_0r_0\varepsilon c(\varepsilon)(\mathcal{K}_1^r -\cos\theta\mathcal{K}_3+\frac{1}{\varepsilon}\mathcal{K}_2
+\mathcal{K}_2\sigma-w(\theta)\mathcal{K}_2)+o(\varepsilon)^{0,\gamma^\prime},\\\nonumber
 \frac{f ^\theta}{\omega} &=&\frac{G\mu C(R)}{2\omega_0^2 R^2}(1-\varepsilon w)(\mathcal{K}_2^\theta+\mathcal{K}_1^\theta+\varepsilon\mathcal{K}_1^\theta\rho)\\\nonumber
 &&\left(1 -\varepsilon w+\frac{3}{2}d(\rho) \right)  +o(\varepsilon)^{0,\gamma^\prime}\\\nonumber
&=&\omega_0r_0\varepsilon c(\varepsilon)(\sin\theta\mathcal{K}_3+\mathcal{K}_1(\rho-w-\frac{3}{2\varepsilon}d(\rho))+o(\varepsilon)^{0,\gamma^\prime}
\\\nonumber
&=&\omega_0r_0\varepsilon c(\varepsilon)(\sin\theta\mathcal{K}_3+\mathcal{K}_1\sigma)+o(\varepsilon)^{0,\gamma^\prime},
 \end{eqnarray}
 where we have also used definitions (\ref{kernels}), the estimates (\ref{bounds}), the observation that constant functions are in the kernel of the operator $\mathcal{K}_1^\theta $ and (\ref{denominator2}) for computing the term of $O(\varepsilon)$  of $d(\rho).$
 From (\ref{forcecomponentorder1}), keeping also into account that $\mathcal{K}_2$ maps constants into constants, we obtain
 \begin{eqnarray}\label{effe}
   \langle\frac{f}{\omega}\rangle \hskip12.5cm \\\nonumber\hskip3cm =\omega_0r_0\varepsilon c(\varepsilon)(
   \langle \cos(\theta)(\mathcal{K}_1^r
+\mathcal{K}_2\sigma-w(\theta)\mathcal{K}_2)-\sin\theta \mathcal{K}_1\sigma \rangle -2\pi\mathcal{K}_3)+o(\varepsilon)^{0,\gamma^\prime}.\hskip1cm\end{eqnarray}
The lemma follows from this (\ref{forcecomponentorder1}) and
\begin{eqnarray}
  \frac{1}{\omega(R+r\cos(\cdot))^3}\frac{1}{\langle\frac{1}{\omega(R+r\cos(\cdot))^3}\rangle}&=& \frac{1}{2\pi}+O(\varepsilon)^{0,\gamma}.
\end{eqnarray}
If   $\phi:\RR\rightarrow\RR$ is a $2\pi-$periodic function we let $\frac{1}{2}\phi_0+\sum_1^\infty(\phi_n^1\cos(n\cdot)+\phi_n^2\sin(n\cdot))$ be the Fourier series of $\phi.$ In particular for the unknowns $\rho$ and $w$ we have:
\begin{eqnarray}\label{structure-ro-w}
\left\{\begin{array}{l}
\rho=\sum_2^\infty(\rho_n^1\cos(n\cdot)+\rho_n^2\sin(n\cdot)),\\
w=\sum_1^\infty(w_n^1\cos(n\cdot)+w_n^2\sin(n\cdot)),
\end{array}\right.
\end{eqnarray}
where we have used the fact that $\rho$ and $w$ have zero average and (\ref{fixcenter}).

\begin{lem}\label{fourierk1r} The function $\mathcal{K}_1^r$ can be represented in the form
\begin{eqnarray}\label{coefficientk1r}
  \mathcal{K}_1^r(\theta) &=& -4\pi\sum_{n=2}^\infty n(\rho_n^1\cos{n\theta}+ \rho_n^2\sin{n\theta})+o(\varepsilon)^{0,\gamma^\prime}.
\end{eqnarray}
\end{lem}
\begin{proof} From (\ref{kernels}) and the discussion in Lemma \ref{mapdefinedbyk} it follows
\begin{eqnarray}
 \mathcal{K}_1^r = \mathcal{K}_1^{r,+}&=&\int_{-\pi}^{\pi}\int_{-\pi}^{\pi}\frac{2\varepsilon^2\rho^+(\alpha+\theta,\theta)}{(z(\beta)+\varepsilon^2 z(\alpha))^\frac{3}{2}}d\alpha d\beta,
\end{eqnarray}
where $\rho^+(\alpha+\theta,\theta)=\frac{1}{2}(\rho(\alpha+\theta)+\rho(-\alpha+\theta)-2\rho(\theta)).$
With the change of variable $z(\beta)=\varepsilon^2\xi^2$, that implies  $\sqrt{1-(\frac{\varepsilon\xi}{2})^2}d\beta=\varepsilon d\xi,$   $\mathcal{K}_1^r$ takes the form
\begin{eqnarray}
  \mathcal{K}_1^r &=& 4\int_0^\pi\int_0^{\frac{2}{\varepsilon}}\frac{2\varepsilon^2\rho^+(\alpha+\theta,\theta)}{(\xi^2+ z(\alpha))^\frac{3}{2}} \frac{d\alpha d\xi}{\sqrt{1-(\frac{\varepsilon\xi}{2})^2}} \\\nonumber
    &=&  8\int_0^\pi\int_0^\infty\frac{\rho^+(\alpha+\theta,\theta)}{(\xi^2+ z(\alpha))^\frac{3}{2}} d\alpha d\xi    -8\int_0^\pi\int_\frac{2}{\varepsilon}^\infty\frac{\rho^+(\alpha+\theta,\theta)}{(\xi^2+ z(\alpha))^\frac{3}{2}} d\alpha d\xi\\\nonumber
    && +8\int_0^\pi\int_0^{\frac{2}{\varepsilon}}\frac{\rho^+(\alpha+\theta,\theta)}{(\xi^2+ z(\alpha))^\frac{3}{2}} \frac{1-\sqrt{1-(\frac{\varepsilon\xi}{2})^2}}{\sqrt{1-(\frac{\varepsilon\xi}{2})^2}}d\alpha d\xi=I^1+I^2+I^3.
\end{eqnarray}
From (ii) and (v) in the proof of Lemma \ref{mapdefinedbyk} we have
\begin{eqnarray}\label{stimaI2}
  \vert I^2\vert &\leq& 8\left\|\rho\right\|_{C^{1,\gamma}}\int_0^\pi\int_\frac{2}{\varepsilon}^\infty\frac{\vert \alpha\vert^{1+\gamma}}{\xi^{-3}}d\xi\leq  C\left\|\rho\right\|_{C^{1,\gamma}}\varepsilon^2,\\\nonumber
  \frac{\vert I^2(\theta_1)-I^2(\theta_2)\vert}{\vert \theta_1-\theta_2\vert^{\nu^\prime}} &\leq& 8\left\|\rho\right\|_{C^{1,\gamma}}\int_0^\pi\int_\frac{2}{\varepsilon}^\infty\frac{\vert \alpha\vert^{(1+\gamma)(1-\nu^\prime)}}{\xi^{-3}}d\xi\leq C\left\|\rho\right\|_{C^{1,\gamma}}\varepsilon^2\\\
  \Rightarrow I^2 &=& o(\varepsilon)^{0,\gamma^\prime},
 \end{eqnarray}
 where to conclude that $I^2$ is $o(\varepsilon)^{0,\gamma^\prime}$ we have used (\ref{stimaI2})$_1$ and (\ref{stimaI2})$_2$ that show $I^2$ is H$\ddot{\rm o}$lder continuous and  the linear dependence of  $I^2$ from $\rho$.
To estimate $I^3$ we set $q(s)=\frac{1-\sqrt{1-(\frac{s}{2})^2}}{\sqrt{1-(\frac{s}{2})^2}}$ and observe that $q(s)\leq C s^2,\hskip.3cm s\in(0,1);\hskip.4cm q(s)\leq C\frac{1}{\sqrt{1-\frac{s}{2}}}, \hskip.3cm s\in(1,2).$ Then we have
\begin{eqnarray}\label{stimaI3}
\hskip1.7cm  \vert I^3\vert &\leq& C\left\|\rho\right\|_{C^{1,\gamma}}\int_0^\pi\int_0^{\frac{2}{\varepsilon}}\frac{\vert \alpha\vert^{1+\gamma}}{(\xi^2+ z(\alpha))^\frac{3}{2}} q(\varepsilon\xi)d\alpha d\xi \\\nonumber
    & &\leq C\left\|\rho\right\|_{C^{1,\gamma}}\int_0^\pi\int_0^{\frac{1}{\varepsilon}}\frac{\vert \alpha\vert^{1+\gamma}}{(\xi^2+ z(\alpha))^\frac{3}{2}} (\varepsilon\xi)^2d\alpha d\xi\\\nonumber
    &&\hskip4cm +C\left\|\rho\right\|_{C^{1,\gamma}}
    \int_0^\pi\int_{\frac{1}{\varepsilon}}^{\frac{2}{\varepsilon}}\frac{\vert \alpha\vert^{1+\gamma}}{(\xi^2+ z(\alpha))^\frac{3}{2}} \frac{1}{\sqrt{1-\frac{\varepsilon\xi}{2}}}d\alpha d\xi
\end{eqnarray} and
\begin{eqnarray}\label{stimaI31}
  \int_0^\pi\int_0^{\frac{1}{\varepsilon}}\frac{\vert \alpha\vert^{1+\gamma}}{(\xi^2+ z(\alpha))^\frac{3}{2}} (\varepsilon\xi)^2d\alpha d\xi&\leq& \varepsilon^2\int_0^\pi\int_0^1\frac{\vert \alpha\vert^{1+\gamma})d\alpha d\xi}{(\xi^2+ z(\alpha))^\frac{3}{2}}\\\nonumber &&\hskip3cm  +C\varepsilon^2\int_1^{\frac{1}{\varepsilon}}\frac{ 1}{\xi } d\xi\hskip2cm\\\nonumber
    &\leq&  C\varepsilon^2(1+\log{\frac{1}{\varepsilon}}),\\\nonumber
    \int_0^\pi\int_{\frac{1}{\varepsilon}}^{\frac{2}{\varepsilon}}\frac{\vert \alpha\vert^{1+\gamma}}{(\xi^2+ z(\alpha))^\frac{3}{2}} \frac{1}{\sqrt{1-\frac{\varepsilon\xi}{2}}}d\alpha d\xi&\leq&C\int_{\frac{1}{\varepsilon}}^{\frac{2}{\varepsilon}}  \frac{1}{\xi^3\sqrt{1-\frac{\varepsilon\xi}{2}}}d\xi\\\nonumber &=& C\varepsilon^2\int_1^2  \frac{1}{s^3\sqrt{1-\frac{s}{2}}}d s.
\end{eqnarray}
Proceeding as in (\ref{stimaI3}) and (\ref{stimaI31}) we also get that $\vert I^3(\theta_1)-I^3(\theta_2)\vert\leq C\varepsilon^2 \log\frac{1}{\varepsilon}\vert \theta_1-\theta_2\vert^{\nu^\prime}.$ From these estimates and the linearity of $I^3$ in $\rho$ we conclude that, as $I^2$, also $I^3$ is $o(\varepsilon)^{0,\gamma^\prime}$. From this and (\ref{stimaI2}), (\ref{stimaI3}) and (\ref{stimaI31}) we conclude that

\begin{eqnarray}\label{similar}
  \mathcal{K}_1^r &=& 8\int_0^\pi\int_0^\infty\frac{\rho^+(\alpha+\theta,\theta)}{(\xi^2+ z(\alpha))^\frac{3}{2}} d\alpha d\xi+ o(\varepsilon)^{0,\gamma^\prime} = I^1+ o(\varepsilon)^{0,\gamma^\prime}.
\end{eqnarray}
To compute the Fourier expansion of $I^1=I^1(\theta)$ we begin by observing that by Fubini theorem and the identity  $\int_0^\infty\frac{a}{(\xi^2+ a)^\frac{3}{2}}d\xi=1,$ valid for $a>0,$ we have
\begin{eqnarray}
  I^1 &=& 8\int_0^\pi \frac{\rho^+(\alpha+\theta,\theta)}{z(\alpha)} d\alpha.
\end{eqnarray}
Therefore the Fourier coefficients $(I^1)_n^i, i=1,2$ of $I^1$ are given by
\begin{eqnarray}\label{coefficientsexpression}
  (I^1)_n^1 &=& \frac{8}{\pi}\int_{-\pi}^\pi\cos{n\theta}(\int_0^\pi \frac{\rho^+(\alpha+\theta,\theta)}{z(\alpha)} d\alpha)d\theta,\hskip.3cm n\geq 0,\\\nonumber
  (I^1)_n^2 &=& \frac{8}{\pi}\int_{-\pi}^\pi\sin{n\theta}(\int_0^\pi \frac{\rho^+(\alpha+\theta,\theta)}{z(\alpha)} d\alpha)d\theta,\hskip.3cm n\geq 1.
\end{eqnarray}
Since we have $\vert\frac{\rho^+(\alpha+\theta,\theta)}{z(\alpha)}\vert\leq \left\|\rho\right\|_{C^{1,\gamma}}\vert\alpha\vert^{\gamma-1}$ the functions $\cos{n\theta}\frac{\rho^+(\alpha+\theta,\theta)}{z(\alpha)}, \sin{n\theta}\frac{\rho^+(\alpha+\theta,\theta)}{z(\alpha)}$ are integrable in $(-\pi,\pi)\times(0,\pi)$ and therefore by Fubini theorem we can interchange the order of integration in (\ref{coefficientsexpression}). From this and
\begin{eqnarray}
  \rho^+(\alpha+\theta,\theta) &=& \sum_{n=2}^\infty(\cos{n\alpha}-1)(\rho_n^1\cos{n\theta}+ \rho_n^2\sin{n\theta})
\end{eqnarray}
we obtain
\begin{eqnarray}
  (I^1)_0^1 &=& 0,\\\nonumber (I^1)_n^i &=& -4\rho_n^i\int_0^\pi\frac{1-\cos{n\alpha}}{1-\cos{\alpha}}d\alpha = -4n\pi\rho_n^i,\hskip.3cm  n\geq 1,\hskip.3cm i=1,2.
\end{eqnarray}
Since $\rho_1^j=0,\,j=1,2$ we also have $(I^1)_1^j=0,\,j=1,2$.
This concludes the proof.
\end{proof}

\begin{lem}\label{fourierk2}Set
\begin{eqnarray}
  \tilde\phi=-\frac{3}{2}
        (\frac{z(\beta)}{D_0^2(\beta,\alpha-\theta)}\phi(\alpha); & &\hat\phi= -\frac{3}{2}\frac{\varepsilon^2 z(\alpha-\theta)}{D_0^2(\beta,\alpha-\theta)}\phi(\alpha),
\end{eqnarray}
and define $\tilde{{\mathcal K}}_j,\;\hat{{\mathcal K}}_j$ by $\tilde{{\mathcal K}}_j\phi:={\mathcal K}\tilde{\phi},\;\hat{{\mathcal K}}_j\phi:={\mathcal K}\hat{\phi}$
Then:
 \begin{eqnarray}\label{coefficientk3}
 (\mathcal{K}_1\phi)(\theta)&=&
 4\pi\sum_1^\infty(\phi_n^2\cos{n\theta}-\phi_n^1\sin{n\theta})+o(\varepsilon)^{0,\gamma^\prime}\hskip.3cm \\\nonumber
 (\tilde{{\mathcal K}}_1\phi)(\theta)&=&-2\pi\sum_1^\infty(\phi_n^2\cos{n\theta}-\phi_n^1\sin{n\theta})\hskip.3cm +o(\varepsilon)^{0,\gamma^\prime},\\\nonumber
 (\hat{{\mathcal K}}_1\phi)(\theta)&=&-4\pi\sum_1^\infty(\phi_n^2\cos{n\theta}-\phi_n^1\sin{n\theta})\hskip.3cm +o(\varepsilon)^{0,\gamma^\prime}.
 \end{eqnarray}
 \begin{eqnarray}\label{coefficientk2}
 (\mathcal{K}_2\phi)(\theta)&=&-2\pi\phi_0+o(\varepsilon)^{0,\gamma^\prime}\hskip.3cm \\\nonumber
 (\tilde{{\mathcal K}}_2\phi)(\theta)&=&\pi\phi_0\hskip.3cm +o(\varepsilon)^{0,\gamma^\prime},\\\nonumber
 (\hat{{\mathcal K}}_2\phi)(\theta)&=& 2\pi\phi_0\hskip.3cm +o(\varepsilon)^{0,\gamma^\prime}.
 \end{eqnarray}
 \end{lem}
\begin{proof} The same arguments used in the proof of Lemma \ref{fourierk1r} to establish (\ref{similar}) lead to
\begin{eqnarray}\label{fourierk1}
 \mathcal{K}_1\phi &=& 4\int_{-\pi}^\pi\int_0^\infty
\frac{\sin\alpha\phi^-(\alpha+\theta)}{(\xi^2+ z(\alpha))^\frac{3}{2}}d\alpha d\xi  +o(\varepsilon)^{0,\gamma^\prime}\\\nonumber
 \mathcal{K}_1\tilde\phi &=& -6\int_{-\pi}^\pi\int_0^\infty
\frac{\xi^2\sin\alpha\phi^-(\alpha+\theta)}{(\xi^2+ z(\alpha))^\frac{5}{2}}d\alpha d\xi  +o(\varepsilon)^{0,\gamma^\prime}\\\nonumber
 \mathcal{K}_1\hat\phi &=& -6\int_{-\pi}^\pi\int_0^\infty
\frac{z(\alpha)\sin\alpha\phi^-(\alpha+\theta)}{(\xi^2+ z(\alpha))^\frac{5}{2}}d\alpha d\xi  +o(\varepsilon)^{0,\gamma^\prime}.
\end{eqnarray}
We sketch the computation for the function $\mathcal{K}_1\tilde\phi.$ We set $z(\beta)=\varepsilon^2\xi^2$ as in the proof of Lemma \ref{fourierk1r} and rewrite $\mathcal{K}_1\tilde\phi$ in the form

\begin{eqnarray}
  \hskip1.7cm\mathcal{K}_1\tilde\phi &=& -6\int_{-\pi}^\pi\int_0^{\frac{2}{\varepsilon}}\frac{\xi^2\sin\alpha\phi^-(\alpha+\theta)}{(\xi^2+ z(\alpha))^\frac{5}{2}} \frac{d\alpha d\xi}{\sqrt{1-(\frac{\varepsilon\xi}{2})^2}} \\\nonumber
    &=& -6\int_{-\pi}^\pi\int_0^\infty\frac{\xi^2\sin\alpha\phi^-(\alpha+\theta)}{(\xi^2+ z(\alpha))^\frac{5}{2}} d\alpha d\xi    +6\int_{-\pi}^\pi\int_\frac{2}{\varepsilon}^\infty\frac{\xi^2\sin\alpha\phi^-(\alpha+\theta)}{(\xi^2+ z(\alpha))^\frac{5}{2}} d\alpha d\xi\\\nonumber
    && -6\int_0^\pi\int_0^{\frac{2}{\varepsilon}}\frac{\xi^2\sin\alpha\phi^-(\alpha+\theta)}{(\xi^2+ z(\alpha))^\frac{5}{2}} q(\varepsilon\xi)d\alpha d\xi=J^1+J^2+J^3,
\end{eqnarray}
where as before $q(s)=\frac{1-\sqrt{1-(\frac{s}{2})^2}}{\sqrt{1-(\frac{s}{2})^2}}.$  We can then estimate $J^2$ and $J^3$  as we have estimated $I^2$ and $I^3$ in Lemma \ref{fourierk1r}.

The same procedure yields
\begin{eqnarray}\label{fourierk22}
 \mathcal{K}_2\phi &=& -2\int_{-\pi}^\pi\int_0^\infty
\frac{z(\alpha)\phi(\alpha+\theta)}{(\xi^2+ z(\alpha))^\frac{3}{2}}d\alpha d\xi  +o(\varepsilon)^{0,\gamma^\prime},\hskip.8cm\\\nonumber
 \mathcal{K}_2\tilde\phi &=& 3\int_{-\pi}^\pi\int_0^\infty
\frac{\xi^2z(\alpha)\phi(\alpha+\theta)}{(\xi^2+ z(\alpha))^\frac{5}{2}}d\alpha d\xi  +o(\varepsilon)^{0,\gamma^\prime}, \\\nonumber
 \mathcal{K}_2\hat\phi &=& 3\int_{-\pi}^\pi\int_0^\infty
\frac{z(\alpha)^2\phi(\alpha+\theta)}{(\xi^2+ z(\alpha))^\frac{5}{2}}d\alpha d\xi  +o(\varepsilon)^{0,\gamma^\prime}.
\end{eqnarray}
The expressions (\ref{coefficientk2}) are a straightforward consequence of (\ref{fourierk22}) and the identities

\begin{eqnarray}\label{identities-int}
\hskip1.7cm\int_0^\infty\frac{a}{(\xi^2+a)^\frac{3}{2}}d\xi=1,
\int_0^\infty\frac{\xi^2a}{(\xi^2+a)^\frac{5}{2}}d\xi=\frac{1}{3},
\int_0^\infty\frac{a^2}{(\xi^2+a)^\frac{5}{2}}d\xi=\frac{2}{3},\hskip.5cm a>0.\hskip.5cm
\end{eqnarray}
From (\ref{fourierk1}) and (\ref{identities-int}) we obtain
\begin{eqnarray}\label{mathcalK1}
\mathcal{K}_1\phi&=&4I+o(\varepsilon)^{0,\gamma^\prime},\\\nonumber  \mathcal{K}_1\tilde\phi&=&-2I+o(\varepsilon)^{0,\gamma^\prime},\\\nonumber \mathcal{K}_1\hat\phi&=&-4I+o(\varepsilon)^{0,\gamma^\prime},
\end{eqnarray}
with
\begin{eqnarray*}
  I &=& \int_{-\pi}^\pi\frac{\sin\alpha}{z(\alpha)}\phi^-(\alpha+\theta)d\alpha,
\end{eqnarray*}
To compute the Fourier coefficients of $I=I(\theta)$ we observe that Fubini theorem implies
\begin{eqnarray}\label{fourierofI}
  I_n^1 &=&\frac{1}{\pi} \int_{-\pi}^\pi\frac{\sin\alpha}{z(\alpha)}
  (\int_{-\pi}^\pi\phi^-(\alpha+\theta)\cos{n\theta}d\theta)d\alpha,\\\nonumber
  I_n^2 &=&\frac{1}{\pi}\int_{-\pi}^\pi\frac{\sin\alpha}{z(\alpha)}
  (\int_{-\pi}^\pi\phi^-(\alpha+\theta)\sin{n\theta}d\theta)d\alpha.
\end{eqnarray}
From (\ref{fourierofI})  and $\phi^-(\alpha+\theta)=\sum_{n=1}^\infty\sin{n\alpha}(-\phi_n^1\sin{n\theta}+\phi_n^2\cos{n\theta})$
it follows
\begin{eqnarray*}
  I_n^1 &=&  \phi_n^2\int_{-\pi}^\pi\frac{\sin\alpha\sin{n\alpha}}{2(1-\cos\alpha)}d\alpha=\pi\phi_n^2,\\\nonumber
  I_n^2 &=& -\phi_n^1\int_{-\pi}^\pi\frac{\sin\alpha\sin{n\alpha}}{2(1-\cos\alpha)}d\alpha=-\pi\phi_n^1.
\end{eqnarray*}
The expressions (\ref{coefficientk3}) follow from this and (\ref{mathcalK1}).
The proof is concluded.
\end{proof}

We are now in the position of deriving explicit expressions of ${\mathcal F}^r$ and ${\mathcal F}^\theta$ in term of the Fourier coefficients of $\rho$ and $w$.
\begin{prop}\label{firststep}
We have
\begin{eqnarray}\label{calf-expressions}
\hskip.5cm\frac{{\mathcal F}^r}{\omega_0r_0} &=&-4\pi c(\varepsilon)\\\nonumber
 &&\hskip.5cm +4\pi\varepsilon c(\varepsilon)[\frac{1}{2}\cos{\theta}+\sum_2^\infty (1-n)(\rho_n^1\cos{n\theta}+\rho_n^2\sin{n\theta})+w(\theta)]+\varepsilon{\mathcal N}_F^r,\\\nonumber
\frac{{\mathcal F}^\theta}{\omega_0r_0} &=& 4\pi\varepsilon c(\varepsilon)[\sum_1^\infty(-w_n^2\cos{n\theta}+w_n^1\sin{n\theta})+\frac{1}{2}\sin{\theta}]+\varepsilon{\mathcal N}_F^\theta.
\end{eqnarray}
where
${\mathcal N}_F^r={\mathcal N}_F^r(\rho,w;\theta,\varepsilon),\;{\mathcal N_F}^\theta={\mathcal N}_F^\theta(\rho,w;\theta,\varepsilon)$ are $o(\varepsilon^0)^{0,\gamma^\prime}$.
\end{prop}
\begin{proof}
From Lemma \ref{fourierk1r} and Lemma \ref{fourierk2} we have
\begin{eqnarray}\label{kr-etc}
{\mathcal K}_1^r+(\frac{1}{\varepsilon}-w) {\mathcal K}_2 &=& -4\pi[\sum_2^\infty n(\rho_n^1\cos{n\theta}+\rho_n^2\sin{n\theta})+\frac{1}{\varepsilon}-w(\theta)]+o(\varepsilon^0)^{0,\gamma^\prime}
\end{eqnarray}
From the definition of $\sigma$ in Lemma \ref{linearpart} and the definition of the operators
$\tilde{{\mathcal K}}_j,\;\hat{{\mathcal K}}_j$ in Lemma \ref{fourierk2} we see that (\ref{coefficientk2}) implies:
\begin{eqnarray}\label{kj-sigma}
{\mathcal K}_1\sigma &=&{\mathcal K}_1(\rho-w)+\tilde{{\mathcal K}}_1(\cos(\cdot)+\cos(\theta))+\hat{{\mathcal K}}_1(\rho+\rho(\theta))\\\nonumber
&=& 4\pi[\sum_1^\infty(-w_n^2\cos{n\theta}+w_n^1\sin{n\theta})+\frac{1}{2}\sin{\theta}]+o(\varepsilon^0)^{0,\gamma^\prime},\\\nonumber
{\mathcal K}_2\sigma &=&{\mathcal K}_2(\rho-w)+\tilde{{\mathcal K}}_2(\cos(\cdot)+\cos(\theta))+\hat{{\mathcal K}}_2(\rho+\rho(\theta))\\\nonumber
&=&4\pi(\rho(\theta)+\frac{1}{2}\cos(\theta))+o(\varepsilon^0)^{0,\gamma^\prime}.
\end{eqnarray}
From (\ref{kr-etc}) and (\ref{kj-sigma}) we obtain
\begin{eqnarray}\label{ex-average}
\langle \cos(\cdot)({\mathcal K}_1^r-w {\mathcal K}_2+{\mathcal K}_2\sigma)-\sin(\cdot){\mathcal K}_1\sigma\rangle &=& o(\varepsilon^0)^{0,\gamma^\prime}.
\end{eqnarray}
Equations (\ref{calf-expressions})$_1$ and (\ref{calf-expressions})$_2$ follow from (\ref{rhslinearpart}) and (\ref{kr-etc}), (\ref{kj-sigma}) and (\ref{ex-average}).
\end{proof}
\begin{rem}\label{effe-expression}
From the estimate (\ref{ex-average}) and (\ref{effe}) it follows
\begin{eqnarray}\label{effe-expression-1}
   \langle\frac{f}{\omega}\rangle &=& -2\pi\omega_0r_0\varepsilon c(\varepsilon) \mathcal{K}_3+o(\varepsilon)^{0,\gamma^\prime}.
\end{eqnarray}
\end{rem}
\section{The proof of Theorem \ref{main}}\label{the-proof}
We let $X$ the set of the pairs $(\rho,w)$ of $2\pi-$periodic functions $\rho\in W^{2,2}(-\pi,\pi),\;w\in  W^{1,2}(-\pi,\pi)$  that satisfy
\begin{eqnarray}\label{x-definition}
\rho_0=0,&&w_0=0,\\\nonumber
\rho_1^j=0,&&j=1,2.
\end{eqnarray}
  $X$ is a Banach space with the norm $\left\|(\rho,w)\right\|_X:=\left\|\rho\right\|_{W^{2,2}}+\left\|w\right\|_{W^{1,2}}$.
We assume throughout that $(\rho,w)$ is bounded by some constant $M$ that will be fixed later:
\begin{eqnarray}\label{m-bound}
\left\|(\rho,w)\right\|_X\leq M.
\end{eqnarray}
We say that a map $h(\rho,w;\theta,\varepsilon)$ that satisfies (\ref{m-bound}) is $O_X(\varepsilon^k)$ if $h$ is such that
 \begin{eqnarray}\label{ssmoothness-x}
   \vert| h(\rho,w)\vert|_{L^2} &\leq& C\varepsilon^k ,\\\nonumber
   \vert| h(\rho_1,w_1)-h(\rho_2,w_2)\vert|_{L^2}&\leq&
   C\varepsilon^k\vert|(\rho_1,w_1)-(\rho_2,w_2)\vert|_X,
 \end{eqnarray}
 for some constant $C>0$ independent of $\varepsilon$.

 If instead (\ref{ssmoothness-x}) holds with a constant $C=C(\varepsilon)$ that   depends on $\varepsilon$ and $\lim_{\varepsilon\rightarrow 0^+} C(\varepsilon)=0$ we say that $h(\rho,w;\alpha,\theta,\varepsilon)=o_X(\varepsilon^k)$.
\begin{rem}\label{ox-oc}
If $(\rho,w)\in X$ then
\begin{eqnarray}\label{ox-oc-1}
 h(\rho,w;\theta,\varepsilon)=O(\varepsilon^k)^{0,\gamma^\prime}&\Rightarrow& h(\rho,w;\theta,\varepsilon)=O_X(\varepsilon^k),\\\nonumber
  h(\rho,w;\theta,\varepsilon)=o(\varepsilon^k)^{0,\gamma^\prime}&\Rightarrow&
  h(\rho,w;\theta,\varepsilon)=o_X(\varepsilon^k)
 \end{eqnarray}
\end{rem}
From the expressions of $\frac{a^r}{\omega},\;\frac{a^\theta}{\omega}$ and (\ref{ro-w-form}) we obtain
\begin{eqnarray}\label{lhs-expressions}
\frac{1}{\omega_0r_0}\frac{a^r}{\omega} &=& -1+\varepsilon(\rho^{\prime\prime}-\rho-w)+\varepsilon{\mathcal N}_a^r,\\\nonumber
\frac{1}{\omega_0r_0}\frac{a^\theta}{\omega} &=& \varepsilon(2\rho^\prime+w^\prime)+\varepsilon{\mathcal N}_a^\theta.
\end{eqnarray}
where ${\mathcal N}_a^r={\mathcal N}_a^r(\rho,w;\theta,\varepsilon)$ and ${\mathcal N}_a^r={\mathcal N}_a^r(\rho,w;\theta,\varepsilon)$ are $o_X(\varepsilon^0)$.
We are now in the position of transforming the equations (\ref{newsystem}) into an infinite set of equations for the Fourier coefficients of the unknowns $\rho=\sum_2^\infty(\rho_n^1\cos{n\theta}+\rho_n^2\sin{n\theta}),\;w=\sum_1^\infty(w_n^1\cos{n\theta}+w_n^2\sin{n\theta})$.

From (\ref{lhs-expressions}) and Proposition \ref{firststep} it follows that we can rewrite (\ref{newsystem}) in the form
\begin{eqnarray}\label{newsystem-explicit}
\hskip1.5cm &&-1+\varepsilon(\rho^{\prime\prime}-\rho-w) =-4\pi c(\varepsilon)\\\nonumber
 && +4\pi\varepsilon c(\varepsilon)[\frac{1}{2}\cos{\theta}+\sum_2^\infty (1-n)(\rho_n^1\cos{n\theta}+\rho_n^2\sin{n\theta})+w(\theta)] +\varepsilon{\mathcal N}^r,\\\nonumber
&&\varepsilon(2\rho^\prime+w^\prime)\\\nonumber &&=4\pi\varepsilon c(\varepsilon)[\sum_1^\infty(-w_n^2\cos{n\theta}+w_n^1\sin{n\theta})+\frac{1}{2}\sin{\theta}] +\varepsilon{\mathcal N}^\theta,
\end{eqnarray}
where we have set ${\mathcal N}^r={\mathcal N}_F^r-{\mathcal N}_a^r$ and ${\mathcal N}^\theta={\mathcal N}_F^\theta-{\mathcal N}_a^\theta$.
 From Remark \ref{ox-oc} and (\ref{lhs-expressions}) and Proposition \ref{firststep} it follows
 that ${\mathcal N}^r$ and ${\mathcal N}^\theta$ are $o_X(\varepsilon^0)$.

 By taking the inner product of (\ref{newsystem-explicit}) with $\frac{1}{\sqrt{2\pi}}$
  we obtain
\begin{eqnarray}\label{c-epsilon}
c(\varepsilon)=\frac{1}{4\pi}-\frac{1}{8\pi}\varepsilon{\mathcal N}^r_0.
\end{eqnarray}
As we have remarked after introducing system (\ref{newsystem}), on the basis of Lemma \ref{can-be-done} we don't need to consider the projection of (\ref{newsystem})$_2$ on the subspace of constant functions.
 Using (\ref{c-epsilon}) we rewrite (\ref{newsystem-explicit}) in the form
\begin{eqnarray}\label{newsystem-explicit-1}
\hskip1.5cm\rho^{\prime\prime}-\rho-w -[\sum_2^\infty (1-n)(\rho_n^1\cos{n\theta}+\rho_n^2\sin{n\theta})+w(\theta)]&=& \frac{1}{2}\cos{\theta}+\tilde{{\mathcal N}}^r,\\\nonumber
2\rho^\prime+w^\prime - [\sum_1^\infty(-w_n^2\cos{n\theta}+w_n^1\sin{n\theta})]&=&  \frac{1}{2}\sin{\theta}+\tilde{{\mathcal N}}^\theta,
\end{eqnarray}
where $\tilde{{\mathcal N}}^r$ and $\tilde{{\mathcal N}}^\theta$ are $=o_X(\varepsilon^0)$ and  $\tilde{{\mathcal N}}^r_0=0$. On the basis of Lemma \ref{can-be-done} and (\ref{newidentity}) we can also assume $\tilde{{\mathcal N}}^\theta_0=0$. Moreover we can replace (\ref{newsystem-explicit-1})$_2$ with its projection on the orthogonal complement of the subspace generated by $\cos\theta$ and $\sin\theta$. Indeed the projection of (\ref{newsystem-explicit-1})$_2$ on this subspace is an automatic consequence of (\ref{newsystem})$_1$. This follows from the identities (\ref{identities}) that imply $(\frac{a^r}{\omega}-{\mathcal F}^r)_1^1=(\frac{a^\theta}{\omega}-{\mathcal F}^\theta)_1^1$ and $(\frac{a^r}{\omega}-{\mathcal F}^r)_1^2=-(\frac{a^\theta}{\omega}-{\mathcal F}^\theta)_1^2$ and from the particular structure of the r.h.s. of (\ref{newsystem})$_2$. Therefore we can assume $(\tilde{{\mathcal N}}^\theta)_1^j=0,\,j=1,2$ and rewrite (\ref{newsystem-explicit-1}) in the form
\begin{eqnarray}\label{newsystem-explicit-2}
\hskip.8cm\sum_2^\infty(-n^2+n-2)(\rho_n^1\cos{n\theta}+\rho_n^2\sin{n\theta})-2\sum_1^\infty
(w_n^1\cos{n\theta}+w_n^2\sin{n\theta})=\tilde{{\mathcal N}}^r,\\\nonumber
\sum_2^\infty 2 n
(\rho_n^2\cos{n\theta}-\rho_n^1\sin{n\theta})+\sum_2^\infty 2 n
(w_n^2\cos{n\theta}-w_n^1\sin{n\theta})=\tilde{{\mathcal N}}^\theta,
\end{eqnarray}
where we have also used the Fourier series of $\rho,\,w,\,\rho^\prime,\,w^\prime,\,\rho^{\prime\prime}$.

Let $Y$ the set of pairs $(p,v)$ of $2\pi-$periodic functions $p,\,v\in L^2(-\pi,\pi)$ that satisfy (\ref{x-definition}) (with $\rho=p,\,w=v$). $Y$ is a Banach space with the norm $\left\|(p,v)\right\|_Y:=\left\|p\right\|_{L^2}+\left\|v\right\|_{L^2}$.
\begin{lem}\label{linear-operator}
Let $L(\rho,w)$ the l.h.s. of (\ref{newsystem-explicit-2}). Then
\begin{description}
\item[(i)] $L(\rho,w)\in Y,\;\;\;(\rho,w)\in X.$
\item[(ii)] The map $L:X\rightarrow Y$ defined by
\begin{eqnarray}\label{linear-operator-1}
X\ni (\rho,w)\rightarrow L(\rho,w)\in Y,
\end{eqnarray} is linear and bounded and has a bounded inverse
 $L^{-1}:Y\rightarrow X$.
\end{description}
\end{lem}
\begin{proof}
By inspecting (\ref{newsystem-explicit-2}) and by observing that $(\rho,w)\in X$ implies that
the series in (\ref{newsystem-explicit-2}) are well defined $L^2$ functions proves {\rm(i)} and also that $L$ is bounded.  To show that $L^{-1}$ exists and is bounded we note that from (\ref{newsystem-explicit-2}) the equation
\begin{eqnarray}\label{solving-linear}
 L(\rho,w)=(p,v),
 \end{eqnarray}
 is equivalent to the system
\begin{eqnarray}\label{w1-w2}
\left\{\begin{array}{l}
 w_1^1=p_1^1,\\
 w_1^2=p_1^2,
\end{array}\right.\hskip1.5cm
\end{eqnarray}

\begin{eqnarray}\label{wn-rhon-1}
\left\{\begin{array}{l}
(2-n+n^2)\rho_n^1+2 w_n^1=-p_n^1,\\
(2-n+n^2)\rho_n^2+2 w_n^2=-p_n^2,
\end{array}\right.,\;\;n\geq 2,
\end{eqnarray}

\begin{eqnarray}\label{wn-rhon-2}
\left\{\begin{array}{l}
2 n\rho_n^2+(n+1) w_n^2=w_n^1,\\
2 n\rho_n^1+(n+1) w_n^1=w_n^2,
\end{array}\right.,\;\;n\geq 2,\hskip.7cm
\end{eqnarray}
where (\ref{w1-w2}) and (\ref{wn-rhon-1}) follow from (\ref{newsystem-explicit-1})$_1$ and (\ref{wn-rhon-2})
from (\ref{newsystem-explicit-1})$_2$.
Equations (\ref{wn-rhon-1}) and (\ref{wn-rhon-2}) imply
\begin{eqnarray}\label{ron-wn-explicit}
\left(\begin{array}{l}
\rho_n^1\\w_n^1
\end{array}\right)=\frac{1}{2-3n+n^3}\left(\begin{array}{lr}
n+1 & -2\\
-2 n & 2-n+n^2
\end{array}\right)\left(\begin{array}{l}
-p_n^1\\v_n^2
\end{array}\right),\\\nonumber
\left(\begin{array}{l}
\rho_n^2\\w_n^2
\end{array}\right)=\frac{1}{2-3n+n^3}\left(\begin{array}{lr}
n+1 & -2\\
-2 n & 2-n+n^2
\end{array}\right)\left(\begin{array}{l}
-p_n^2\\v_n^1
\end{array}\right),
\end{eqnarray}
and therefore
\begin{eqnarray}\label{ro-w-estimate}
\left\{\begin{array}{l}
\vert\rho_n^j\vert\leq\frac{C}{n^2}\sum_{i=1,2}(\vert p_n^i\vert+\frac{1}{n}\vert v_n^i\vert),\\\\
\vert w_n^j\vert\leq\frac{C}{n}\sum_{i=1,2}(\frac{1}{n}\vert p_n^i\vert+\vert v_n^i\vert),
\end{array}\right.,\;\;n\geq 2.
\end{eqnarray}
These inequalities and (\ref{w1-w2}) show that the Fourier coefficients $\rho_n^j,\,j=1,2;\,n=2,\cdots$ and  $w_n^j,\,j=1,2;\,n=1,\cdots$  define functions $\rho\in W^{2,2}$ and $w\in W^{1,2}$ that satisfy (\ref{x-definition}).
\end{proof}
From Lemma \ref{linear-operator} it follows that (\ref{newsystem-explicit-2}) is equivalent
to the equation
\begin{eqnarray}\label{fixed-point-equation}
(\rho,w)=(0,\bar{w})+L^{-1}\left(\begin{array}{l}
\tilde{{\mathcal N}}^r(\rho,w)\\\tilde{{\mathcal N}}^\theta(\rho,w)
\end{array}\right):=G(\rho,w),
\end{eqnarray}
where we have set $\bar{w}=-\frac{1}{4}\cos\theta$. Therefore the problem of solving (\ref{newsystem-explicit-2}) is reduced to the existence of a fixed point for the map $G:X\rightarrow X$.

Fix $M=2\left\|(0,\bar{w})\right\|_X$ in (\ref{m-bound}) and let $\overline{X}=\{(\rho,w)\in X: \left|(\rho,w)\right\|_X\leq M\}$. Then from the fact that $\tilde{{\mathcal N}}^r$ and $\tilde{{\mathcal N}}^\theta$ are $o_X(\varepsilon^0)$ and from Lemma \ref{linear-operator} we have:
\begin{eqnarray}\label{g-in-x}
(\rho,w)\in\overline{X}\Rightarrow\hskip2cm &&\\\nonumber \left\|G(\rho,w)\right\|_X &\leq&\left|(0,\bar{w})\right\|_X+C(\varepsilon),\\\nonumber
\left\|G(\rho,w)-G(\tilde{\rho},\tilde{w})\right\|_X &\leq& C(\varepsilon) \left\|(\rho,w)-(\tilde{\rho},\tilde{w})\right\|_X.
\end{eqnarray}
Therefore, for $\varepsilon>0$ smaller than some $\varepsilon_0>0$, $G:\overline{X}\rightarrow \overline{X}$ is a contraction and (\ref{fixed-point-equation}) has a unique solution $(\rho^*,w^*)\in \overline{X}$. Due to the equivalence between (\ref{fixed-point-equation}) and (\ref{newsystem-explicit-2}), $(\rho^*,w^*)$ is a solution of (\ref{equilibrium}), (\ref{continuity1}).  Moreover (\ref{g-in-x}) imply:
\begin{eqnarray}\label{rho-w-bounds-fix}
\left\|\rho^*\right\|_{W^{2,2}}\leq C(\varepsilon),\\\nonumber
\left\|w^*-\bar{w}\right\|_{W^{1,2}}\leq C(\varepsilon),
\end{eqnarray}
where $C(\varepsilon)\rightarrow 0$ as $\varepsilon\rightarrow 0^+$. From these estimates, Lemma \ref{integrands} and Lemma \ref{mapdefinedbyk} it follows that the r.h.s. of system (\ref{system}) computed for $(\rho,w)=(\rho^*,w^*)$ is of class $C^{0,\gamma}$. Therefore we can regard (\ref{system}) as a system of ode with a H$\ddot{\rm o}$lder continuous r.h.s. This shows that $\rho^*\in C^{2,\gamma}$ and $w^*\in C^{1,\gamma}$. The uniqueness in {\rm(iii)} is just a restatement of the fact that that $(\rho^*,w^*)$ is the unique fixed point of a contraction on $\overline{X}$. The properties of $s$ in {\rm(iv)} follow (\ref{continuity1}) and from {\rm(i)} and {\rm(ii)} recalling also the expression of $c(\varepsilon)$ in Lemma \ref{linearpart} and (\ref{c-epsilon}). The statement on the smoothness of $\Omega$ follows from (\ref{jeieomega})$_2$ and {\rm(i)} and {\rm(ii)}. The expression of $\overline{\Omega}$ follows from (\ref{jeieomega})$_2$ and from Remarks \ref{effe-expression-1} and \ref{k3-remark}. The proof of Theorem \ref{main} is concluded.

\section*{Acknowledgment}
It is a pleasure for G. Fusco and P. Negrini to thank the CAMGSD, Department of Mathematics at IST for the kind hospitality during the preparation of the present work.
\nocite{*}

\bibliographystyle{plain}

\begin{thebibliography}{99}

 \bibitem{newton9}I.~ Newton.
 {\em Philosophiae naturalis principia mathematica.} Londini: Jussi Societatus Regiae ac typis Joseph Streater, 1687.

  \bibitem{chandra2} Chandrasekhar, S.
 {\em Ellipsoidal figures of equilibrium.} New York: Dover, 1987.


 \bibitem{oliva2} W.~M.~Oliva.
 \newblock  Massas Fluidas em Rotação e os Elipsóides de Riemann - uma Abordagem Informal.
 \newblock {\em Rev. Matem\'atica Universit\'aria} {\bf 43} (2007), pp.~28--37.

 \bibitem{riemann} B.~Riemann.
\newblock   Ein beitrag zu den unterschungen über die bewengung eines flüssigen gleichartigen ellipsoides.
\newblock {\em Abhandlungen der K$\ddot{o}$niglichen Gesellschaft der Wissenschaften zu G$\ddot{o}$ttingen} {\bf 9} (1860),~pp. 3--36.

 \bibitem{dirichlet} J.~P.~G.~L.~Dirichlet.
 \newblock Untersuchungen über ein problem der hydrodynamik. Aus dessen nachlass hergestellt von R. Dedekind.
 \newblock {\em Abhandlungen der Königlichen Gesellschaft der Wissenschaften zu Göttingen} {\bf 8}, (1858/1859) pp.~225--264.

\bibitem{p}
H.~Poincar\'e.
\newblock Sur l'\'equilibre d'une masse fluide anim\'ee d'un mouvement de rotation.
\newblock {\em Acta\ mathematica} {\bf 7} (1885), pp.~259--380.

\bibitem{sofie} S.~Kowalewsky.
 \newblock Zus$\ddot{\rm a}$tze und Bemerkungen zu Laplace's Untersuchung $\ddot{\rm u}$ber die Gestalt der Saturnsringe.
\newblock {\em  Astronomiche Nachrichten} (1885), pp.~38--48.

\bibitem{a}
J.~F.~G.~Auchmuty.
\newblock  Existence of axisymmetric equilibrium figures.
\newblock {\em Arch.\ Rat.\ Mech.\ Analysis} {\bf 65} (1977), pp.~249--261 .

\bibitem{ab}
J.~F.~G.~Auchmuty and R.~Beals.
\newblock  Variational solutions of some nonlinear free-boundary problems.
\newblock {\em Arch.\ Rat.\ Mech.\ Analysis} {\bf 43} (1971), pp.~255--271 .


\bibitem{cf}
L.~A.~Caffarelli and A.~Friedman.
\newblock The shape of axisymmetric rotating fluid.
\newblock {\em J.\ Funct.\ Analysis} {\bf 35} (1980), pp.~109--142.


\bibitem{fno}
G.~Fusco, P.~Negrini and W.~M.~Oliva.
\newblock Stationary motion of a thin self-gravitating toroidal incompressible liquid layer.
\newblock {\em Preprint} (2011).


\bibitem{miranda}C.~Miranda.
\newblock Sulle propriet\`a di regolarit\`a di certe trasformazioni integrali.
\newblock {\em Mem.\ Accademia dei Lincei} {\bf VII}, Sez.I$^{a}$ (1965), pp.~303--336.
 \end{thebibliography}

\end{document}